\documentclass[onecolumn,11pt]{IEEEtran}
\IEEEoverridecommandlockouts
\usepackage{silence}
\usepackage{soul}
\usepackage[mathscr]{eucal}
\usepackage[cmex10]{amsmath}
\usepackage{epsfig,epsf,psfrag}
\usepackage{amssymb,amsmath,amsthm,amsfonts,latexsym}
\usepackage{amsmath,graphicx,bm,xcolor,url,overpic}
\usepackage[caption=false]{subfig} 
\usepackage{fixltx2e}
\usepackage{array}
\usepackage{verbatim}
\usepackage{bm}
\usepackage{algorithmic}
\usepackage{algorithm}
\usepackage{verbatim}
\usepackage{textcomp}
\usepackage{mathrsfs}
\usepackage{epstopdf}

\newcommand{\openone}{\leavevmode\hbox{\small1\normalsize\kern-.33em1}}

\catcode`~=11 \def\UrlSpecials{\do\~{\kern -.15em\lower .7ex\hbox{~}\kern .04em}} \catcode`~=13 

\allowdisplaybreaks[4]

\newcommand{\nn}{\nonumber}

\newcommand{\calA}{\mathcal{A}}
\newcommand{\calB}{\mathcal{B}}
\newcommand{\calC}{\mathcal{C}}
\newcommand{\calD}{\mathcal{D}}
\newcommand{\calE}{\mathcal{E}}
\newcommand{\calF}{\mathcal{F}}
\newcommand{\calG}{\mathcal{G}}
\newcommand{\calH}{\mathcal{H}}

\newcommand{\calJ}{\mathcal{J}}
\newcommand{\calK}{\mathcal{K}}
\newcommand{\calL}{\mathcal{L}}
\newcommand{\calM}{\mathcal{M}}
\newcommand{\calN}{\mathcal{N}}

\newcommand{\calP}{\mathcal{P}}

\newcommand{\calR}{\mathcal{R}}
\newcommand{\calS}{\mathcal{S}}
\newcommand{\calT}{\mathcal{T}}
\newcommand{\calU}{\mathcal{U}}
\newcommand{\calV}{\mathcal{V}}

\newcommand{\calZ}{\mathcal{Z}}

\newcommand{\ba}{\mathbf{a}}
\newcommand{\bA}{\mathbf{A}}
\newcommand{\bb}{\mathbf{b}}
\newcommand{\bB}{\mathbf{B}}

\newcommand{\bd}{\mathbf{d}}

\newcommand{\be}{\mathbf{e}}

\newcommand{\bg}{\mathbf{g}}
\newcommand{\bG}{\mathbf{G}}
\newcommand{\bh}{\mathbf{h}}
\newcommand{\bH}{\mathbf{H}}

\newcommand{\bN}{\mathbf{N}}
\newcommand{\bo}{\mathbf{o}}

\newcommand{\br}{\mathbf{r}}

\newcommand{\bt}{\mathbf{t}}

\newcommand{\bu}{\mathbf{u}}
\newcommand{\bU}{\mathbf{U}}
\newcommand{\bv}{\mathbf{v}}
\newcommand{\bV}{\mathbf{V}}
\newcommand{\bw}{\mathbf{w}}
\newcommand{\bW}{\mathbf{W}}
\newcommand{\bx}{\mathbf{x}}
\newcommand{\bX}{\mathbf{X}}

\newcommand{\bY}{\mathbf{Y}}
\newcommand{\bz}{\mathbf{z}}
\newcommand{\bZ}{\mathbf{Z}}

\newcommand{\rmE}{\mathrm{E}}

\newcommand{\rmN}{\mathrm{N}}

\newcommand{\rmQ}{\mathrm{Q}}

\newcommand{\bbE}{\mathsf{E}}

\newcommand{\bbI}{\openone}

\newcommand{\bbN}{\mathbb{N}}

\newcommand{\bbR}{\mathbb{R}}
\newcommand{\bbS}{\mathbb{S}}

\DeclareMathAlphabet{\mathbsf}{OT1}{cmss}{bx}{n}
\DeclareMathAlphabet{\mathssf}{OT1}{cmss}{m}{sl}

\newcommand{\rvA}{\mathsf{A}}

\newcommand{\rvR}{\mathsf{R}}

\newcommand{\rvT}{\mathsf{T}}

\DeclareSymbolFont{bsfletters}{OT1}{cmss}{bx}{n}  
\DeclareSymbolFont{ssfletters}{OT1}{cmss}{m}{n}
\DeclareMathSymbol{\bsfGamma}{0}{bsfletters}{'000}
\DeclareMathSymbol{\ssfGamma}{0}{ssfletters}{'000}
\DeclareMathSymbol{\bsfDelta}{0}{bsfletters}{'001}
\DeclareMathSymbol{\ssfDelta}{0}{ssfletters}{'001}
\DeclareMathSymbol{\bsfTheta}{0}{bsfletters}{'002}
\DeclareMathSymbol{\ssfTheta}{0}{ssfletters}{'002}
\DeclareMathSymbol{\bsfLambda}{0}{bsfletters}{'003}
\DeclareMathSymbol{\ssfLambda}{0}{ssfletters}{'003}
\DeclareMathSymbol{\bsfXi}{0}{bsfletters}{'004}
\DeclareMathSymbol{\ssfXi}{0}{ssfletters}{'004}
\DeclareMathSymbol{\bsfPi}{0}{bsfletters}{'005}
\DeclareMathSymbol{\ssfPi}{0}{ssfletters}{'005}
\DeclareMathSymbol{\bsfSigma}{0}{bsfletters}{'006}
\DeclareMathSymbol{\ssfSigma}{0}{ssfletters}{'006}
\DeclareMathSymbol{\bsfUpsilon}{0}{bsfletters}{'007}
\DeclareMathSymbol{\ssfUpsilon}{0}{ssfletters}{'007}
\DeclareMathSymbol{\bsfPhi}{0}{bsfletters}{'010}
\DeclareMathSymbol{\ssfPhi}{0}{ssfletters}{'010}
\DeclareMathSymbol{\bsfPsi}{0}{bsfletters}{'011}
\DeclareMathSymbol{\ssfPsi}{0}{ssfletters}{'011}
\DeclareMathSymbol{\bsfOmega}{0}{bsfletters}{'012}
\DeclareMathSymbol{\ssfOmega}{0}{ssfletters}{'012}

\newcommand{\hatx}{\hat{x}}
\newcommand{\hatX}{\hat{X}}

\newcommand{\haty}{\hat{y}}
\newcommand{\hatY}{\hat{Y}}

\newcommand{\veps}{\varepsilon}

\newcommand{\bomega}{\bm{\omega}}

\newcommand{\blambda}{\bm{\lambda}}

\newcommand{\bmu}{\bm{\mu}}

\newcommand{\bSigma	}{\bm{\Sigma}}

\DeclareMathOperator{\diag}{diag}

\DeclareMathOperator{\tr}{tr}

\DeclareMathOperator{\var}{\mathsf{Var}}

\DeclareMathOperator{\cov}{\mathsf{Cov}}

\DeclareMathOperator{\rank}{rank}

\newcommand{\bzero}{\mathbf{0}}
\newcommand{\bone}{\mathbf{1}}

\newtheorem{theorem}{Theorem} 
\newtheorem{lemma}[theorem]{Lemma}

\newtheorem{proposition}[theorem]{Proposition}
\newtheorem{corollary}[theorem]{Corollary}
\newtheorem{definition}{Definition}

\newtheorem{gwlemma}{Lemma}

\newtheorem{gwproposition}{Proposition}

\usepackage{graphicx,cite}
\usepackage{epstopdf}
\usepackage{enumerate}
\usepackage{color}
\usepackage{bbm}
\usepackage{booktabs}
\usepackage{makecell}
\definecolor{Dyellow}{RGB}{254,152,0}
\definecolor{Dgreen}{RGB}{0,176,80}

\makeatletter

\newcommand{\Rmnum}[1]{\expandafter\@slowromancap\romannumeral #1@}
\makeatother
\def\BibTeX{{\rm B\kern-.05em{\sc i\kern-.025em b}\kern-.08em
T\kern-.1667em\lower.7ex\hbox{E}\kern-.125emX}}
\usepackage[hidelinks]{hyperref}
\usepackage{mathtools,microtype,float,needspace}
\usepackage{tikz,pgfplots,pgfplotstable}
\usetikzlibrary{arrows.meta}
\pgfplotsset{compat=1.18}

\newcommand{\norm}[1]{\left\lVert#1\right\rVert}
\newcommand{\ip}[2]{\left\langle#1,#2\right\rangle}
\allowdisplaybreaks[0]

\hypersetup{pdftitle={Exact Second-Order Asymptotics for the Gaussian Gray-Wyner Problem},
pdfauthor={Lin Zhou, Lei Yu},pdfsubject={Gaussian Gray-Wyner manuscript, version 5}}

\begin{document}
\title{Exact Second-Order Asymptotics for the Gaussian Gray--Wyner Problem}
\author{Lin Zhou and Lei Yu
\thanks{Lin Zhou is with the School of Automation and Intelligent Manufacturing, Southern University of Science and Technology, Shenzhen, China (e-mail: zhoul9@sustech.edu.cn).}
\thanks{Lei Yu is with the School of Statistics and Data Science, Nankai University, Tianjin, China (e-mail: leiyu@nankai.edu.cn).}}
\maketitle
\begin{abstract}
We revisit the Gray-Wyner problem of lossy compression of bivariate Gaussian sources under the quadratic distortion measure and derive the exact second-order asymptotics. Specifically, we first identify the boundary rate tuples using the explicit first-order asymptotics by Yu (TIT 2023) and then characterize the tradeoff among three second-order coding rates subject to a constraint on the joint excess-distortion probability. Our main contribution is the achievability proof, where we derive a covariance shell covering lemma that is analogous to the type covering lemma for the discrete case. In particular, our covering lemma classifies the source sequences into power shell types via their empirical covariance matrix. For each such power shell type, we derive sufficient conditions on the compression rates so that covering within distortion levels are achievable. Finally, asymptotical analysis via the Berry-Esseen theorem yields the desired result.
\end{abstract}

\section{Introduction}
\label{sec:intro}

Data compression, a.k.a., source coding, was pioneered by Shannon~\cite{shannon1948} in this seminal paper. Specifically, Shannon proved that the entropy is the minimal compression rate for lossless compression of a discrete memoryless source (DMS) with vanishing error probability as the blocklength tends to infinity. A decade later, Strassen~\cite{strassen1962} refined Shannon's result by deriving the second-order asymptotics that approximates the finite blocklength performance of optimal codes by characterizing the order and the coefficient of the second dominant term of the non-asymptotic compression rate as a function of the blocklength and a given error probability. Recently, the result of Strassen was further refined by Kontoyiannis and Verd\'u who derived further refined third-order asymptotics~\cite{kontoyiannisverdu2014}. As demonstrated in~\cite{strassen1962},~\cite{kontoyiannisverdu2014}, refined asymptotics approximates the performance of optimal codes at finite blocklengths.

Although insightful, lossless compression is valid only for a DMS. For continuous memoryless sources, lossy source coding, a.k.a., rate-distortion was introduced by Shannon~\cite{shannon1959}. In lossy source coding, the reconstruction is no longer required to be exactly the same as the original source sequence and the imperfection is measured via a distortion function. Shannon~\cite{shannon1959} derived the first-order asymptotics for lossy source coding under the expected distortion criterion when the blocklength tends to infinity and named the coefficient of the first dominant term of the minimal compression rate of optimal codes as the rate-distortion function. The same result also holds when one uses a vanishing excess-distortion probability criterion. Several decades later, the second-order asymptotics of lossy source coding was derived by Ingber and Kochman~\cite{ingberkochman2011} and Kostina and Verd\'u~\cite{kostinaverdu2012}. Furthermore, for an equiprobable DMS under Hamming distortion, Kostina and Verd\'u also characterized the third-order asymptotics~\cite[Theorem~27]{kostinaverdu2012}.

Beyond point-to-point (P2P) data compression, motivated by practical applications in the storage/retrieval of correlated sources~\cite[Section~V-C]{viswanathaakyolrose2014} and cache-aided delivery of correlated content~\cite{hassanzadeh2020}, Gray and Wyner~\cite{graywyner1974} proposed a source coding problem in which a pair of correlated source sequences is compressed, in either a lossless or lossy manner, and transmitted from a single transmitter over a common link and two private links to two different receivers. The first-order asymptotics, a.k.a., the rate region for the lossless case and the rate-distortion region for the lossy case, was characterized in~\cite{graywyner1974}. Ever since, the problem is known as the Gray-Wyner problem, a classical multiterminal source coding problem concerning compression of correlated memoryless sources. Several decades later, for the DMS, the second-order asymptotics for the Gray-Wyner problem was derived by Watanabe~\cite{watanabe2017} for the lossless case and by Zhou, Tan and Motani~\cite{zhoutanmotani2017gw} for the lossy case. Furthermore, Liu et al.~\cite[Theorem~17]{liu2020} derived second-order converse bounds for both discrete correlated sources under bounded distortion measures and bivariate Gaussian sources under the quadratic distortion measure. Recently, for bivariate Gaussian sources, the explicit form of the rate-distortion region was completely characterized by Yu~\cite[Theorem~3 and Corollary~2]{yu2023}. Throughout the paper, we refer to the Gray-Wyner problem for lossy compression of bivariate Gaussian memoryless sources under the quadratic distortion measure as the Gaussian Gray-Wyner problem.

In contrast to the discrete case, the exact second-order asymptotics for the Gaussian Gray-Wyner problem remains open until now. The key challenge lies in how to accurately extract the common information from bivariate Gaussian sources and derive the Gaussian analogue of type covering lemma~\cite[Lemma~8]{zhoutanmotani2017gw}. The method of types~\cite{csiszar1998}, tailed and used extensively for the DMS, cannot be applied to continuous sources like Gaussian sources studied in this paper.

\subsection{Main Contributions}
\label{sec:contributions}

We address the above challenge and derive exact second-order asymptotics for the Gaussian Gray-Wyner problem. Our main contribution is an achievability bound that matches the existing converse bound of Liu et al.~\cite[Theorem~17]{liu2020}. In particular, we derive a covariance shell covering lemma that is parallel to the type covering lemma for the DMS. To do so, we first define Gaussian type classes by grouping pairs of correlated Gaussian source sequences according to their empirical covariance matrix, generalizing the power type class for a single Gaussian source in the successive refinement problem~\cite[Section~VI-A]{zhoutanmotani2017sr}. Subsequently, for each Gaussian type class, we use random covering for the common message and use Euclidean ball covering for the private messages. In particular, we show that the achievable common rate equals the minimum common rate derived from the first-order asymptotics when the source distribution is given by the Gaussian type class and two private rates are set as desired values. Finally, we use the Berry--Esseen theorem~\cite[Theorem~1.1]{bentkus2005} to conclude the achievability analysis. For the converse part, we combine the second-order converse bounds for the Gaussian Gray--Wyner problem of Liu et al.~\cite[Theorem~17]{liu2020} and the P2P Gaussian lossy source coding problem of Kostina and Verd\'u~\cite[Theorem~40]{kostinaverdu2012}.

Concurrent with our study, Cao~\cite[Theorem~18]{cao2026image} presented a second-order characterization for the bivariate Gaussian Gray--Wyner problem studied in this paper. However, we would like to clarify that no detailed proof steps were provided for~\cite[Theorem~18]{cao2026image} nor did the author of~\cite[Theorem~18]{cao2026image} presented the result exactly as in our paper despite some similarity. In particular, as emphasized above, our paper provides a complementary and explicit second-order asymptotic analysis of the Gaussian Gray-Wyner problem for any first-order boundary rate tuples with positive common rate. Specifically, we classify the boundary rate tuples into four regimes and derive closed-form expressions for the supporting weights and dispersion. For the achievability, we establish a uniform covariance-shell covering lemma that guarantees simultaneous covering of every source block in a covariance shell using with explicit rate bounds. For the converse part, we identify the entropy-support quantity and its covariance gradient needed to specialize the result of Liu et al.~\cite[Theorem~17]{liu2020}, and handle zero supporting weights through the P2P Gaussian converse of~\cite[Theorem~40]{kostinaverdu2012}.

\subsection{Other Related Works}
\label{sec:related}

Besides the Gray--Wyner problem, the second-order asymptotics has also been established for other multiterminal source coding problems. Specifically, Tan and Kosut~\cite{tankosut2014} studied the Slepian--Wolf problem~\cite{slepianwolf1973} of lossless compression of correlated sources, which was generalized by Nomura and Han~\cite{nomurahan2014} to mixed sources. Furthermore, No, Ingber, and Weissman~\cite{noingberweissman2016} studied the successive refinement problem~\cite{rimoldi1994} under separate excess-distortion probabilities criteria while by Zhou, Tan, and Motani~\cite{zhoutanmotani2017sr} considered the joint excess-distortion probability criterion. For the Wyner-Ziv problem~\cite{wynerziv1976} of lossy source coding with decoder side information, achievability result was obtained by Watanabe, Kuzuoka, and Tan~\cite{watanabekuzuokatan2015}, which was recently refined by Li and Li~\cite{lili2025}. Very recently, exact second-order asymptotics was characterized by Zhou~\cite{zhouwz2026} under mild conditions. For a comprehensive summary of second-order asymptotics of lossy source coding, refer to the monograph~\cite{zhoumotani2023}.

The Gray--Wyner problem is also closely related to the study of common information. This line of research was pioneered by Wyner~\cite{wyner1975}, who characterized the minimum common rate for the lossless Gray-Wyner problem when the sum rate of three encoders equals the joint entropy of the two correlated sources. Subsequently, Viswanatha, Akyol, and Rose~\cite{viswanathaakyolrose2014} extended~\cite{wyner1975} to the lossy case and explicitly calculated the lossy common information. For the Gaussian Gray-Wyner problem, Sula and Gastpar~\cite{sulagastpar2020} determined the minimum common rate subject to the sum of two private rates.

\subsection{Organization for the Rest of the Paper}
\label{sec:organization}

In Section II, we formulate the Gaussian Gray-Wyner problem, define the first- and second-order asymptotics and recall existing result. In Section III, we present and discuss our main results. The achievability and converse proofs of our main result are presented in Sections IV and V, respectively. Finally, Section VI concludes the paper and discusses future directions. For smooth presentation, the proofs of supporting lemmas are deferred to the appendices.

\subsection*{Notation}
Upper-case letters denote random variables, and lower-case letters denote their realizations. Bold upper-case letters denote random vectors and matrices, while bold lower-case letters denote deterministic vectors and matrices. We use $(\bbN,\bbR,\bbR_+,\bbR_{+,0})$ to denote the sets of positive integers, real numbers, positive real numbers and non-negative real numbers, respectively. For any real number $x\in\bbR$, we define $[x]^+:=\max\{x,0\}$. For any two integers $(a,b)\in\bbN^2$ such that $a\leq b$, we define $[a:b]:=\{a,a+1,\ldots,b\}$ and use $[a]$ to denote $[1:a]$ when $a\geq1$. Given $d\in\bbN$, $\bmu\in\bbR^d$ and $\bSigma\in\bbS_+^d$, we use $\calN(\bmu,\bSigma)$ to denote the $d$-dimensional Gaussian distribution with mean vector $\bmu$ and covariance matrix $\bSigma$. We use $\Phi$ to denote the cumulative distribution function (cdf) for the standard Gaussian random variable with mean zero and variance one and use $\rmQ:=1-\Phi$ to denote its complementary cdf. We write $(\bbE_Q[\cdot],\var_Q[\cdot],\cov_Q(\cdot,\cdot))$ for expectation, variance, and covariance with respect to a distribution $Q$. We use the standard asymptotic notation such as $(O(\cdot),o(\cdot))$. All logarithms are natural. We use the scalar symbols $(R_j,L_j,D_j,M_j,V)$ for rates, second-order coding rates, distortions, alphabet sizes, and dispersion, respectively.

Fix an integer $d\in\bbN$. We use $(\bbS^d,\bbS_+^d,\bbS_{++}^d)$ to denote the real symmetric, positive-semidefinite, and positive-definite $d\times d$ matrices. Fix any two matrices $(\ba,\bb)\in(\bbS^d)^2$, we write $\ba\preceq\bb$ when $\bb-\ba\in\bbS_+^d$, use $|\ba|$ to denote the determinant, $\tr(\ba)$ to denote the trace, $\ba^{\rvT}$ to denote the transpose, $\norm{\ba}=\sqrt{\ip{\ba}{\ba}}$ to denote the Frobenius norm, and define $\ip{\ba}{\bb}:=\tr(\ba^{\rvT}\bb)$. For each $j\in[d]$, we use $\be_j\in\bbR^d$ to denote the $d$-dimensional vector with one being the $j$-th element and zero elsewhere. For a vector $\bmu\in\bbR^d$ and a real number $r\in\bbR_+$, we use $\calB_d(\bmu,r):=\{\bx\in\bbR^d:\norm{\bx-\bmu}\leq r\}$ to denote the $d$-dimensional ball with center $\bmu$ and radius $r$ and define $\calB_d(\bmu,0):=\{\bmu\}$. For each $j\in[d]$, given a vector $\bv\in\bbR^d$, we use $[\bv]_j$ to denote its $j$-th coordinate, and given another integer $n\in\bbN$ and a matrix $\bx^n$ of dimension $d\times n$, we use $(\bx^n)_j$ to denote its $j$-th row. We write $\bzero_d:=(0,\ldots,0)^{\rvT}\in\bbR^d$ and $\bzero_{d\times d}\in\bbR^{d\times d}$ for the all zero vector and all zero matrix, respectively, and use $\bone_d$ for the $d\times d$ identity matrix.

\section{Problem Formulation and Existing Results}
\label{sec:formulation}
\subsection{System Model and Definitions}
\label{sec:model}
\begin{figure}[t]
\centering
\begin{tikzpicture}[>=Latex,line width=.6pt,font=\fontsize{10}{12}\selectfont]
\node[draw,minimum width=1.4cm,minimum height=.65cm] (enc1) at (0,1.1) {$f_1$};
\node[draw,minimum width=1.4cm,minimum height=.65cm] (enc0) at (0,0) {$f_0$};
\node[draw,minimum width=1.4cm,minimum height=.65cm] (enc2) at (0,-1.1) {$f_2$};
\node[draw,minimum width=1.75cm,minimum height=.9cm,align=center] (d1) at (4,1.1) {Decoder 1\\$\phi_1$};
\node[draw,minimum width=1.75cm,minimum height=.9cm,align=center] (d2) at (4,-1.1) {Decoder 2\\$\phi_2$};
\draw[->] (-2.4,0) node[left] {$(X^n,Y^n)$} -- (enc0.west);
\draw[->] (-1.25,0) |- (enc1.west);
\draw[->] (-1.25,0) |- (enc2.west);
\draw (enc0.east) -- (4,0);
\node[fill=white,inner sep=1.5pt] at (2.1,.25) {$S_0\in[M_0]$};
\fill (4,0) circle(1.4pt);
\draw[->] (4,0) -- (d1.south);
\draw[->] (4,0) -- (d2.north);
\draw[->] (enc1.east) -- (d1.west) node[midway,above=.10cm] {$S_1\in[M_1]$};
\draw[->] (enc2.east) -- (d2.west) node[midway,above=.10cm] {$S_2\in[M_2]$};
\draw[->] (d1.east) -- (6.1,1.1) node[right] {$\hatX^n$};
\draw[->] (d2.east) -- (6.1,-1.1) node[right] {$\hatY^n$};
\end{tikzpicture}
\caption{System model for the Gaussian Gray--Wyner problem.}
\label{fig:system-model}
\end{figure}
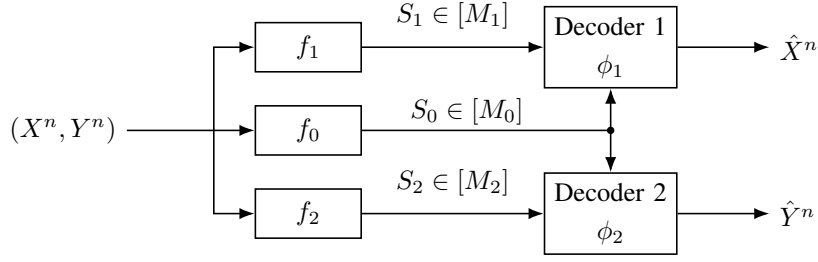
Fix positive real numbers $(\sigma_X,\sigma_Y,\rho_{XY})\in\bbR_+^2\times(0,1)$, two distortion levels $(D_1,D_2)\in(0,\sigma_X^2)\times(0,\sigma_Y^2)$ and four positive integers $(n,M_0,M_1,M_2)\in\bbN^3$. Define the joint source distribution $P_{XY}$:
\begin{align}
P_{XY}&:=\calN(\bzero_2,\bSigma_{XY}),~
\bSigma_{XY}:=
\begin{bmatrix}
\sigma_X^2&\rho_{XY}\sigma_X\sigma_Y\\
\rho_{XY}\sigma_X\sigma_Y&\sigma_Y^2
\end{bmatrix}
\in\bbS_{++}^2.\label{eq:source-cov}
\end{align}
Let $(P_X,P_Y)$ be the induced marginal distributions. We consider memoryless sources $(X^n,Y^n)\in(\bbR^n)^2$ that is generated i.i.d. from $P_{XY}$. As shown in Fig.~\ref{fig:system-model}, in the Gaussian Gray-Wyner problem, three encoders compress a pair of correlated sequences $(X^n,Y^n)$ via three encoders $(f_0,f_1,f_2)$, which are used by two decoders $(\phi_1,\phi_2)$ to generate source sequence estimates $(\hatX^n,\hatY^n)$, respectively. In particular, encoder $f_0$ sends the common message $S_0$ to both decoders while encoder $f_i$, for each $i\in[2]$, sends a private message $S_i$ to decoder $\phi_i$. To evaluate the reconstruction quality, we use the quadratic distortion measure. Specifically, given any $(x,\hat x)\in\bbR^2$ and $(x^n,\hat x^n)\in(\bbR^n)^2$, the quadratic distortion measure and its block extension are defined as follows:
\begin{align}
d(x,\hat x)&:=(x-\hat x)^2,
\\
d(x^n,\hat x^n)&:=\frac1n\sum_{i\in[n]}d(x_i,\hat x_i)=\frac1n\norm{x^n-\hat x^n}^2.
\end{align}
We can now define a code that mathematically defines the behavior of the encoders and decoders.
\begin{definition}
\label{def:code}
An $(n,M_0,M_1,M_2)$-code consists of three encoders $f_i$, $i\in[0:2]$, and two decoders $\phi_j$, $j\in[2]$ such that
\begin{align}
f_i:(\bbR^n)^2&\longrightarrow[M_i],
\\
\phi_j:[M_0]\times[M_j]&\longrightarrow\bbR^n.
\end{align}
Given an $(n,M_0,M_1,M_2)$-code and a pair of sequences $\bx^n=(x^n,y^n)\in\bbR^{2\times n}$, the reconstructed sequences are
\begin{align}
\hatx^n&:=\phi_1(f_0(\bx^n),f_1(\bx^n)),\qquad
\haty^n:=\phi_2(f_0(\bx^n),f_2(\bx^n)).
\end{align}
\end{definition}
To define the performance criterion, define the following excess-distortion event:
\begin{align}
\calE_n&:=\{\bx^n\in\bbR^{2\times n}:d(x^n,\hatx^n)>D_1,\text{ or }d(y^n,\haty^n)>D_2\}, .
\end{align}
Thus, $\calE_n^c=\bbR^{2\times n}\setminus\calE_n$ is the non-excess-distortion event. Define the joint excess-distortion probability as
\begin{align}
\veps_n(D_1,D_2)&:=P_{XY}^n(\calE_n)=1-P_{XY}^n(\calE_n^c).\label{eq:Pe}
\end{align}
For ease of notation, we use $\veps_n$ to denote $\veps_n(D_1,D_2)$. The first-order asymptotics, a.k.a., the rate-distortion region, is defined as follows.
\begin{definition}
\label{def:first-order}
A rate tuple $(R_0,R_1,R_2)\in\bbR_{+,0}^3$ is said to be first-order achievable if there exists a sequence of $(n,M_0,M_1,M_2)$ codes such that for each $j\in[0:2]$,
\begin{align}
\limsup_{n\to\infty}\frac{\log M_{j,n}}n&\leq R_j,\qquad\lim_{n\to\infty}\veps_n=0.
\end{align}
The convex closure of all first-order achievable rate tuples is the first-order rate region, denoted by $\calR(D_1,D_2|P_{XY})$.
\end{definition}

In existing studies including the original paper of Gray and Wyner~\cite{graywyner1974}, the performance criterion is the following expected distortion criteria instead of the vanishing excess-distortion probability:

\begin{align}
\limsup_{n\to\infty}\bbE_{P_{XY}^n}[d(X^n,\hatX^n)]&\leq D_1,
\\
\limsup_{n\to\infty}\bbE_{P_{XY}^n}[d(Y^n,\hatY^n)]&\leq D_2.
\end{align}

Denote the corresponding first-order rate region under the expected distortion criterion as $\calR_{\rmE}(D_1,D_2|P_{XY})$. The following lemma shows that the two distortion criteria lead to the same first-order rate region for bivariate Gaussian source.
\begin{lemma}
\label{lem:first-order-criteria}
It follows that
\begin{align}
\calR(D_1,D_2|P_{XY})&=\calR_{\rmE}(D_1,D_2|P_{XY})=\calR^*(D_1,D_2|P_{XY}),
\end{align}
where
\begin{align}
\nn&\calR^*(D_1,D_2|P_{XY})\\*
&:=\bigcup_{\substack{P_{W|XY},\,P_{\hatX|XW},\,P_{\hatY|YW}:\\
\bbE[(X-\hatX)^2]\leq D_1,\\ \bbE[(Y-\hatY)^2]\leq D_2}}
\Big\{(R_0,R_1,R_2)\in\bbR_{+,0}^3: R_0\geq I(X,Y;W),\ R_1\geq I(X;\hatX|W),\ R_2\geq I(Y;\hatY|W)\Big\}.
\end{align}
\end{lemma}
\begin{proof}
Gray and Wyner~\cite{graywyner1974} showed that 1) $\calR^*(D_1,D_2|P_{XY})\subseteq\calR(D_1,D_2|P_{XY})$; 2) $\calR(D_1,D_2|P_{XY})\subseteq\calR_{\rmE}(D_1,D_2|P_{XY})$; 3) $\calR_{\rmE}(D_1,D_2|P_{XY})\subseteq\calR^*(D_1,D_2|P_{XY})$. Thus, these three sets are identical.
\end{proof}
Given private rates $(R_1,R_2)\in\bbR_{+,0}^2$, define the minimal common rate as
\begin{align}
\label{def:mincommon}
\rvR_0(R_1,R_2,D_1,D_2|P_{XY})
&:=\inf\big\{r_0\in\bbR_{+,0}:(r_0,R_1,R_2)\in\calR(D_1,D_2|P_{XY})\big\}.
\end{align}
In the next subsection, we introduce the explicit expressions of $\rvR_0(R_1,R_2,D_1,D_2|P_{XY})$ and $\calR(D_1,D_2|P_{XY})$ for the Gaussian Gray-Wyner problem.

\subsection{Existing First-Order Asymptotics}
\label{sec:cov-first}

Fix a covariance matrix $\bSigma=[\Sigma_{jk}]_{(j,k)\in[2]^2}\in\bbS_{++}^2$ with $\Sigma_{jj}>D_j$ and $\Sigma_{12}>0$. Let $Q_{XY}^{\bSigma}:=\calN(\bzero_2,\bSigma)$ be the two-dimensional Gaussian distribution and let $(Q_X^{\bSigma},Q_Y^{\bSigma})$ be the marginal distributions. The Gaussian rate-distortion functions~\cite[Theorem~10.3.2, Eq.~(10.24)]{coverthomas2006} are
\begin{align}
R_X(Q_X^{\bSigma},D_1)&=\frac12\log\frac{\Sigma_{11}}{D_1},\label{eq:RG-S}\\
R_Y(Q_Y^{\bSigma},D_2)&=\frac12\log\frac{\Sigma_{22}}{D_2}.\label{eq:RY-S}
\end{align}
Furthermore, define the correlation coefficient: 
\begin{align}
\rho(\bSigma)&:=\frac{\Sigma_{12}}{\sqrt{\Sigma_{11}\Sigma_{22}}},\label{eq:rho-S} 
\end{align}
Fix two non-negative real numbers $(r_1,r_2)\in\bbR_{+,0}^2$. Define 
\begin{align}
\alpha(\bSigma;r_1,D_1)&:=\big[R_X(Q_X^{\bSigma},D_1)-r_1\big]^+,\label{eq:alpha-S}\\
\beta(\bSigma;r_2,D_2)&:=\big[R_Y(Q_Y^{\bSigma},D_2)-r_2\big]^+.\label{eq:beta-S}
\end{align} 
Given a non-negative real number $t\in\bbR_{+,0}$, define
\begin{align}
\nu_t&:=\sqrt{1-e^{-2t}}.\label{eq:nu-def}
\end{align}
Given any $(\rho,\alpha,\beta)\in(0,1)\times\bbR_{+,0}^2$, define
\begin{align}
r_{\alpha,\beta}&:=(\rho-\nu_\alpha\nu_\beta)e^{\alpha+\beta}.\label{eq:rho-ab}
\end{align}
Finally, define the following four sets:
\begin{align}
\calP_1&:=\{(\alpha,\beta)\in\bbR_+^2:\nu_\alpha \nu_\beta\leq\rho\leq\min\{\nu_\beta/\nu_\alpha,\nu_\alpha/\nu_\beta\}\},\label{eq:D1}\\
\calP_2&:=\{(\alpha,\beta)\in\bbR_+^2:\rho\leq \nu_\alpha \nu_\beta\},
\\
\calP_3&:=\{(\alpha,\beta)\in\bbR_+^2:\rho>\nu_\beta/\nu_\alpha\},
\\
\calP_4&:=\{(\alpha,\beta)\in\bbR_+^2:\rho>\nu_\alpha/\nu_\beta\},\label{eq:D4}
\end{align}
and the following functions:
\begin{align}
\label{eq:Upsilon}
\Upsilon_\rho(\alpha,\beta)&:=
\begin{cases}
\alpha+\beta-\dfrac12\log\dfrac{1-r_{\alpha,\beta}^2}{1-\rho^2},&\text{if }(\alpha,\beta)\in\calP_1,\\
\alpha+\beta+\dfrac12\log(1-\rho^2),&\text{if }(\alpha,\beta)\in\calP_2,\\
\alpha,&\text{if }(\alpha,\beta)\in\calP_3,\\
\beta,&\text{if }(\alpha,\beta)\in\calP_4
\end{cases}.
\end{align}
When $(\alpha,\beta)\in\bbR_{+,0}^2$ satisfies $\alpha\beta=0$, define
\begin{align}
\label{eq:Upsilon-axes}
\Upsilon_\rho(\alpha,\beta)&:=\max\{\alpha,\beta\}.
\end{align}
The following lemma is obtained~\cite[Corollary~2, Theorem~3, Eq.~(8)]{yu2023} by rescaling the source parameters to unit variance. Recall that $\rho_{XY}$ is the correlation parameter for the source distribution $P_{XY}$ in \eqref{eq:source-cov}.
\begin{lemma}
\label{prop:first-order-selection}
For any $(R_1,R_2)\in\bbR_{+,0}^2$, the minimum common rate satisfies
\begin{align}
\label{eq:R0-yu}
\rvR_0(R_1,R_2,D_1,D_2|P_{XY})&=\Upsilon_{\rho_{XY}}\big([R_X(P_X,D_1)-R_1]^+,[R_Y(P_Y,D_2)-R_2]^+\big).
\end{align}
\end{lemma}

It follows from Lemma~\ref{prop:first-order-selection} that the first-order rate region in Definition~\ref{def:first-order} is
\begin{align}
\label{eq:first-order-epigraph}
\calR(D_1,D_2|P_{XY})
&=\big\{(R_0,R_1,R_2)\in\bbR_{+,0}^3:
R_0\geq\Upsilon_{\rho_{XY}}\big([R_X(P_X,D_1)-R_1]^+,[R_Y(P_Y,D_2)-R_2]^+\big)\big\}.
\end{align}
For each fixed private rate pair $(R_1,R_2) \in \bbR_{+,0}^2$, equality in~\eqref{eq:first-order-epigraph} determines the corresponding minimum common rate.

\subsection{Definitions of Second-Order Asymptotics}
\label{sec:second-order-definitions}
\begin{figure}[t]
\centering
\includegraphics{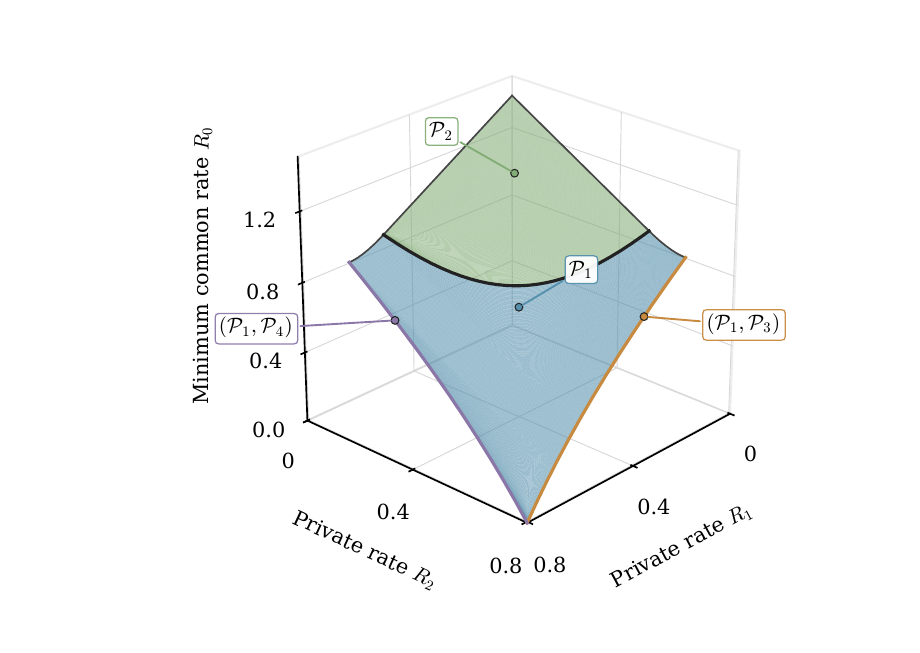}
\caption{Boundary rate tuples with positive common rate when $\sigma_X^2=\sigma_Y^2=1$, $D_1=D_2=0.2$, and $\rho_{XY}=0.6$. The blue and green patches correspond to $(\calP_1,\calP_2)$ and the boundary of the two regions are the black curve. The orange and purple curves mark the $(\calP_1,\calP_3)$ and $(\calP_1,\calP_4)$ interfaces, respectively.}
\label{fig:first-order-regimes}
\end{figure}

To define second-order asymptotics, we need to first specify the boundary rate tuples of the first-order rate-region.
\begin{definition}
\label{def:boundary}
A first-order rate tuple $\br^*=(R_0^*,R_1^*,R_2^*)\in\calR(D_1,D_2|P_{XY})\subset\bbR_{+,0}^3$ is a boundary rate tuple if it is Pareto-minimal in the sense that reducing any one of the three rates would make the rate tuple not in the first-order rate region anymore. Specifically, for each $j\in[0:2]$ and any $0<t\leq R_j^*$,
\begin{align}
\br^*-t\be_{j+1}&\notin\calR(D_1,D_2|P_{XY}).
\end{align}
\end{definition}
Recall the source covariance matrix $\bSigma_{XY}$ in~\eqref{eq:source-cov}, whose diagonal entries are $[\bSigma_{XY}]_{11}=\sigma_X^2$ and $[\bSigma_{XY}]_{22}=\sigma_Y^2$. For a first-order rate tuple $\br^*=(R_0^*,R_1^*,R_2^*)\in\calR(D_1,D_2|P_{XY})\subset\bbR_{+,0}^3$, define
\begin{align}
(\alpha^*,\beta^*)&:=\big(R_X(P_X,D_1)-R_1^*,R_Y(P_Y,D_2)-R_2^*\big).\label{eq:alpha-star}
\end{align}
The following lemma gives exact conditions for a rate tuple to be a boundary rate tuple.
\begin{lemma}
\label{lem:boundary-classification}
A first-order rate tuple $\br^*=(R_0^*,R_1^*,R_2^*)$ with $R_0^*>0$ is a boundary rate tuple if and only if $R_0^*=\Upsilon_{\rho_{XY}}(\alpha^*,\beta^*)$ and $(\alpha^*,\beta^*)\in\bbR_+^2$ satisfies one of the following four conditions:
\begin{enumerate}[(a)]
\item {$(\alpha^*,\beta^*)\in\calP_1$, which implies}
\begin{align}
R_0^*+R_1^*+R_2^*&=R_X(P_X,D_1)+R_Y(P_Y,D_2)-\frac12\log\frac{1-r_{\alpha^*,\beta^*}^2}{1-\rho_{XY}^2}.
\end{align}
\item {$(\alpha^*,\beta^*)\in\calP_2$, which implies}
\begin{align}
R_0^*+R_1^*+R_2^*&=R_X(P_X,D_1)+R_Y(P_Y,D_2)+\frac12\log(1-\rho_{XY}^2).
\end{align}
\item {$(\alpha^*,\beta^*)\in\calP_3$ and $R_2^*=0$, which implies $R_0^*+R_1^*=R_X(P_X,D_1)$.}
\item {$(\alpha^*,\beta^*)\in\calP_4$ and $R_1^*=0$, which implies $R_0^*+R_2^*=R_Y(P_Y,D_2)$.}
\end{enumerate}
\end{lemma}
The proof of Lemma~\ref{lem:boundary-classification} is provided in Appendix~\ref{app:boundary-classification}. Lemma~\ref{lem:boundary-classification} implies that every boundary rate tuple with positive common rate belongs to one of the above four regimes. In Fig.~\ref{fig:first-order-regimes}, we illustrate the boundary rate tuples for a symmetric source with equal distortion levels.

We can now define the second-order asymptotics, which will be characterized in this paper.
\begin{definition}
\label{def:second}
Fix any boundary rate tuple $\br^*=(R_0^*,R_1^*,R_2^*)$ with $R_0^*>0$ and $\veps\in(0,1)$. A tuple $(L_0,L_1,L_2)\in\bbR^3$ is said to be second-order $(\br^*,D_1,D_2,\veps)$-achievable if there exists a sequence of $(n,M_0,M_1,M_2)$-codes such that for each $j\in[0:2]$,
\begin{align}
\limsup_{n\to\infty}\frac{\log M_{j,n}-nR_j^*}{\sqrt n}&\leq L_j,\label{eq:second-rates}
\end{align}
and
\begin{align}
\limsup_{n\to\infty}\veps_n&\leq\veps.\label{eq:second-error}
\end{align}
The closure of all second-order $(\br^*,D_1,D_2,\veps)$-achievable tuples is called the $(\br^*,D_1,D_2,\veps)$-achievable second-order region and denoted by $\calL(\br^*,D_1,D_2,\veps)$.
\end{definition}

Finally, given a boundary rate tuple $\br^*=(R_0^*,R_1^*,R_2^*)$, define the index set for zero private rates as follows:
{\begin{align}
\label{eq:zero-index-set}
\calZ(\br^*)&:=\{j\in[2]:R_j^*=0\}.
\end{align}}
For $j\in\calZ(\br^*)$, we would require $L_j\ge 0$ to guarantee that the number of each private message is strictly positive.

\section{Second-Order Asymptotics}
\label{sec:main}
\subsection{Preliminaries}
\label{sec:critical-preliminaries}

Fix a boundary rate tuple $\br^*=(R_0^*,R_1^*,R_2^*)$ with $R_0^*>0$. Recall the definitions of $\rvR_0(R_1,R_2,D_1,D_2|P_{XY})$ in \eqref{def:mincommon}, $\Upsilon_\rho(\alpha,\beta)$ in \eqref{eq:Upsilon} and $(\alpha^*,\beta^*)$ in~\eqref{eq:alpha-star}. The following lemma proves the differentiability of the minimal common rate.
\begin{proposition}
\label{ass:regularity}\label{prop:regularity}
The following results hold.
\begin{enumerate}
\item The function $(\rho,\alpha,\beta) \mapsto \Upsilon_\rho(\alpha,\beta)$ is continuously differentiable on $(0,1)\times \bbR_{+}^2$.
\item For every boundary rate tuple $\br^*$ with $R_0^*>0$, the function $(\textcolor{black}{\bSigma},r_1,r_2)\mapsto \rvR_0(r_1,r_2,D_1,D_2|Q_{XY}^{\textcolor{black}{\bSigma}})$ is continuously differentiable in the neighborhood of $(\textcolor{black}{\bSigma_{XY}},R_1^*,R_2^*)$ in $\bbS_{++}^2\times\bbR^2$. When either private rate equals zero, the derivative is defined from right (above zero).
\end{enumerate}
\end{proposition}
The proof of Proposition~\ref{prop:regularity} is provided in Appendix~\ref{app:regularity}.

Now, let the derivative of $\Upsilon_\rho(\alpha,\beta)$ be
\begin{equation}\label{eq:Psi-gradient}
(\eta,\lambda_1^*,\lambda_2^*)=\left.\bigl(\partial_\rho\Upsilon_\rho(\alpha,\beta),\partial_\alpha\Upsilon_\rho(\alpha,\beta),\partial_\beta\Upsilon_\rho(\alpha,\beta)\bigr)\right|_{(\rho,\alpha,\beta)=(\rho_{XY},\alpha^*,\beta^*)}.
\end{equation}
It follows from Proposition~\ref{prop:regularity} and Lemma \ref{lem:boundary-classification} that
\begin{align}
\bg&:=\left.\nabla_{\textcolor{black}{\bSigma}}\rvR_0(r_1,r_2,D_1,D_2|Q_{XY}^{\textcolor{black}{\bSigma}})\right|_{(\textcolor{black}{\bSigma},r_1,r_2)=(\textcolor{black}{\bSigma_{XY}},R_1^*,R_2^*)}\label{def:bg}\\
&=\frac12\begin{bmatrix}
(\lambda_1^*-\rho_{XY}\eta)/\sigma_X^2&\eta/(\sigma_X\sigma_Y)\\
\eta/(\sigma_X\sigma_Y)&(\lambda_2^*-\rho_{XY}\eta)/\sigma_Y^2
\end{bmatrix}.\label{eq:G-gradient}
\end{align}
\begin{equation}
g_i:=\left.\partial_{r_i}\rvR_0(r_1,r_2,D_1,D_2|Q_{XY}^{\textcolor{black}{\bSigma}})\right|_{(\textcolor{black}{\bSigma},r_1,r_2)=(\textcolor{black}{\bSigma_{XY}},R_1^*,R_2^*)}=-\lambda_i^*,\qquad i\in[2].\label{eq:private-rate-gradient}
\end{equation}
Let $\blambda^*:=(\lambda_1^*,\lambda_2^*)\in[0,1]^2$. For any $(x,y)\in\bbR^2$, let $\bx:=(x,y)^{\rvT}$ and define the following information density:
\begin{align}
\label{eq:Gaussian-influence-parallel}
\jmath(x,y|R_1^*,R_2^*,D_1,D_2,P_{XY})
&:=R_0^*+\bx^{\rvT}\bg\bx-\tr(\bg\bSigma_{XY}).
\end{align}
We remark that $\jmath(x,y|R_1^*,R_2^*,D_1,D_2,P_{XY})$ is the distortion-tilted information density for the Gaussian Gray-Wyner problem, analogous to the discrete case in \cite[Definition~4, Eqs.~(24)--(26)]{zhoutanmotani2017gw}. Furthermore, define the dispersion, the variance of the distortion-tilted information density, as follows:
\begin{align}
\label{eq:V-def}
V(R_1^*,R_2^*,D_1,D_2|P_{XY})
&:=\var_{P_{XY}}[\jmath(X,Y|R_1^*,R_2^*,D_1,D_2,P_{XY})].
\end{align}
For ease of notation, we use $(\jmath(x,y),V)$ to denote $(\jmath(x,y|R_1^*,R_2^*,D_1,D_2,P_{XY}),V(R_1^*,R_2^*,D_1,D_2|P_{XY}))$. Recall the definitions of $\nu_\cdot$ in \eqref{eq:nu-def}. Define
\begin{align}
u&:=\frac{\nu_{\beta^*}}{\nu_{\alpha^*}}\exp\bigl(\beta^*-\alpha^*\bigr)=\sqrt{\frac{1-e^{-2\beta^*}}{1-e^{-2\alpha^*}}}\exp\bigl(\beta^*-\alpha^*\bigr).\label{eq:u-ratio}
\end{align}
The following lemma presents properties of the distortion-tilted information density and the dispersion function.
\begin{lemma}
\label{lem:critical-quantities}
For the fixed boundary tuple $\br^*$ with $R_0^*>0$, the following claims hold.
\begin{enumerate}
\item The distortion-tilted information density satisfies
\begin{align}
\bbE_{P_{XY}}[\jmath(X,Y)]&=R_0^*.
\end{align}
\item The dispersion function satisfies
\begin{align}
\label{eq:V-slopes}
V&=2\tr[(\bg\bSigma_{XY})^2]
=\frac{1+(\lambda_1^*+\lambda_2^*-1)^2}{2}.
\end{align}
In particular, the dispersion function has the following explicit form:
\begin{align}
\label{eq:V-branches}
V&=
\begin{cases}
0.5\left[1+\left(\dfrac{(u-r)(1-ru)}{u(1-r^2)}\right)^2\right],
&\text{if }(\alpha^*,\beta^*)\in\calP_1,\\
1,&\text{if }(\alpha^*,\beta^*)\in\calP_2,\\
0.5,&\text{if }(\alpha^*,\beta^*)\in\calP_3,\\
0.5,&\text{if }(\alpha^*,\beta^*)\in\calP_4
\end{cases}.
\end{align}
\end{enumerate}
\end{lemma}
The proof of Lemma~\ref{lem:critical-quantities} is provided in Appendix~\ref{app:information-quantity}.

\subsection{Main Results and Discussions}
\label{sec:main-theorem}

Recall that the source distribution is $P_{XY}=\calN(\bzero_2,\bSigma_{XY})$. Furthermore, recall the definitions of the second-order rate region $\calL(\br^*,D_1,D_2,\veps)$ in Definition~\ref{def:second} and the zero private rate index set $\calZ(\br^*)$ in~\eqref{eq:zero-index-set}. Fix any boundary rate tuple $\br^*=(R_0^*,R_1^*,R_2^*)$ with $R_0^*>0$.
\begin{theorem}
\label{thm:main}
Given any $\veps\in(0,1)$, the second-order rate region satisfies
\begin{align}
\label{eq:smooth-region}
\calL(\br^*,D_1,D_2,\veps)&=\Big\{(L_0,L_1,L_2)\in\bbR^3:
L_0+\lambda_1^*L_1+\lambda_2^*L_2\geq\sqrt{V}\,\rmQ^{-1}(\veps),\quad
L_j\geq0\ \forall j\in\calZ(\br^*)\Big\}.
\end{align}
\end{theorem}
The achievability and converse proofs of Theorem~\ref{thm:main} are provided in Sections~\ref{sec:ach} and~\ref{sec:converse}, respectively. We provide a proof sketch here. In the achievability part, we group correlated source sequences into power shell types according to their empirical covariance and derive a covering lemma that bounds the common rate required to cover each power shell type with desired private rates. Specifically, the covariance shell covering lemma enables us to apply the Berry-Esseen theorem to obtain the desired second-order asymptotics. In the converse part, we use the existing result from \cite[Theorem 17]{liu2020} when $(\lambda_1^*,\lambda_2^*)$ are both positive and use the converse result of the P2P case~\cite[Theorem 40]{kostinaverdu2012} when either $\lambda_1^*$ or $\lambda_2^*$ equals zero. However, the application of \cite[Theorem 17]{liu2020} requires one to relate the entropy support expression in \cite[Theorem 17]{liu2020} with the distortion-tilted information density and the dispersion function.

Theorem \ref{thm:main} shows that the second-order rate region for the Gaussian Gray-Wyner problem has exactly the same form as the discrete counterpart~\cite[Theorem~4]{zhoutanmotani2017gw}. In the following corollary, we explicitly calculate the second-order rate region for the Gaussian Gray-Wyner problem. Recall the definitions of $r_\cdot$ in \eqref{eq:rho-ab}, $(\alpha^*,\beta^*)$ in~\eqref{eq:alpha-star}, the four sets $(\calP_1,\ldots,\calP_4)$ in~\eqref{eq:D1}--\eqref{eq:D4}, and $u$ in~\eqref{eq:u-ratio}. For ease of notation, we use $r$ to denote $r_{\alpha^*,\beta^*}$. 
\begin{corollary}
\label{cor:branchwise-region}
The second-order rate region consists of rate tuples $(L_0,L_1,L_2)$ such that $L_j\geq0$ for $j\in\calZ(\br^*)$ and has the following explicit form depending on parameters $(\alpha^*,\beta^*)$ that are defined via the private rates $(R_1^*,R_2^*)$ of the boundary first-order rate tuple $r^*=(R_0^*,R_1^*,R_2^*)$.
\begin{enumerate}
\item If $(\alpha^*,\beta^*)\in\calP_1$,
\begin{align}
\label{eq:branchwise-curved}
L_0+\frac{1-ru}{1-r^2}L_1+\frac{1-r/u}{1-r^2}L_2
&\geq\sqrt V\,\rmQ^{-1}(\veps).
\end{align}
\item If $(\alpha^*,\beta^*)\in\calP_2$,
\begin{align}
\label{eq:branchwise-sum}
L_0+L_1+L_2&\geq\rmQ^{-1}(\veps).
\end{align}
\item If $(\alpha^*,\beta^*)\in\calP_3$,
\begin{align}
\label{eq:branchwise-one-source-X}
L_0+L_1&\geq\frac{\rmQ^{-1}(\veps)}{\sqrt2}.
\end{align}
\item If $(\alpha^*,\beta^*)\in\calP_4$,
\begin{align}
\label{eq:branchwise-one-source-Y}
L_0+L_2&\geq\frac{\rmQ^{-1}(\veps)}{\sqrt2}.
\end{align}
\end{enumerate}
The neighboring conditions agree at every finite interface. In particular, the one-source expressions also hold at their adjoining $\calP_1$ interfaces.
\end{corollary}
The proof of Corollary~\ref{cor:branchwise-region} is provided in Appendix~\ref{app:branchwise-region}.

\section{Achievability Proof}
\label{sec:ach}
This section presents the achievability proof of Theorem \ref{thm:main}. In particular, we derive a covariance shell covering lemma that gives bounds on common and private rates of a deterministic code that covers the correlated source sequences within desired distortion levels at two decoders. This lemma is of independent interest for other multiterminal lossy source coding problems concerning correlated Gaussian sources. 

\subsection{Preliminary Definitions}
\label{sec:preliminaries}
Fix any integer $n\in\bbN$ and $\bx^n=(x^n,y^n)\in(\bbR^n)^2$. Define the following empirical second-moment matrix:
\begin{align}
\label{eq:emp-cov}
\hat{\bSigma}(\bx^n)&:=\frac1n\bx^n(\bx^n)^{\rvT}.
\end{align}
We remark that the diagonal elements of $\hat{\bSigma}(\bx^n)$ are empirical powers $\norm{x^n}^2/n$ and $\norm{y^n}^2/n$. In what follows, we will use $\hat{\bSigma}$ to denote the random covariance matrix $\hat{\bSigma}(\bX^n)$ when $\bX^n=(X^n,Y^n)$ denotes the random source sequences. Recall from \eqref{eq:source-cov} that the covariance matrix is $\bSigma_{XY}$, the standard deviations are $(\sigma_X,\sigma_Y)$ and $\rho_{XY}$ is the correlation parameter of the source distribution. Let $\lambda_{\min}(\bSigma_{XY})$ be its smallest eigenvalue. Define the following two constants
\begin{align}
c_{\mathrm{cor}}&:=\sigma_X^{-2}+\sigma_Y^{-2}+\frac{2}{\sigma_X\sigma_Y},\\
\eta_{\mathrm{cov}}&:=\frac18\min\left\{\lambda_{\min}(\bSigma_{XY}),\sigma_X^2-D_1,\sigma_Y^2-D_2,\frac{\rho_{XY}}{c_{\mathrm{cor}}},\frac{1-\rho_{XY}}{c_{\mathrm{cor}}}\right\}.\label{eq:explicit-covariance-radius}
\end{align}
Furthermore, define the two sets:
\begin{align}
\calK_0&:=\{\bSigma\in\bbS^2:\norm{\bSigma-\bSigma_{XY}}\leq\eta_{\mathrm{cov}}\},\label{eq:inner-covariance-neighborhood}
\end{align}
\begin{align}
\calK&:=\{\bSigma\in\bbS^2:\norm{\bSigma-\bSigma_{XY}}\leq2\eta_{\mathrm{cov}}\}.\label{eq:covariance-neighborhoods}
\end{align}
Given a real number $t\in\bbR$, define the quantization function $\mathfrak{r}_n(t):=\lfloor nt+1/2\rfloor/n$. Given any matrix $\bSigma\in\bbS^2$, define its quantized version by
\begin{align}
\label{eq:covariance-quantizer}
\mathfrak{q}_n(\bSigma)&:=
\begin{bmatrix}
\mathfrak{r}_n(\Sigma_{11})&\mathfrak{r}_n(\Sigma_{12})\\
\mathfrak{r}_n(\Sigma_{12})&\mathfrak{r}_n(\Sigma_{22}).
\end{bmatrix}
\end{align}
Let
\begin{align}
\label{def:quansetK0}
\mathfrak{q}_n(\calK_0)
:=\{\mathfrak{q}_n(\bSigma):~\Sigma\in\calK_0\}
\end{align}
be the set of distinct quantized versions for the matrices in the set $\calK_0$ and let $N_n:=|\mathfrak{q}_n(\calK_0)|$ denote its size. We also write $\mathfrak{q}_n(\calK_0)$ as $(\Sigma_1,\ldots,\Sigma_{N_n})$. For each $\Sigma\in\calK_0$,
\begin{align}
\norm{\mathfrak{q}_n(\bSigma)-\bSigma}^2
&=|\mathfrak{r}_n(\Sigma_{11})-\Sigma_{11}|^2+2|\mathfrak{r}_n(\Sigma_{12})-\Sigma_{12}|^2+|\mathfrak{r}_n(\Sigma_{22})-\Sigma_{22}|^2,
\\
&\leq\frac{1+2+1}{4n^2}=\frac1{n^2},\label{eq:quantizer-error}
\end{align}
where each scalar rounding error is at most $\frac1{2n}$. Thus, $\norm{\mathfrak{q}_n(\bSigma)-\bSigma}\leq1/n$. The triangle inequality implies $\mathfrak{q}_n(\calK_0)\subset\calK$ when $\frac{1}{n}\leq\eta_{\mathrm{cov}}$. Since each entry of $\mathfrak{q}_n(\bSigma)$ and $\mathfrak{q}_n(\bSigma)$ differs by at most $\frac{1}{n}$ for any $\Sigma\in\calK_0$, there are at most $\lfloor2n\eta_{\mathrm{cov}}\rfloor+3$ different quantization values in the range of $2\eta_{\mathrm{cov}}$. As a result,
\begin{equation}\label{eq:quantized-center-count}
N_n\leq(\lfloor2n\eta_{\mathrm{cov}}\rfloor+3)^3\leq(2\eta_{\mathrm{cov}}+3)^3n^3.
\end{equation}

Recall that $\bx^n=(x^n,y^n)\in(\bbR^n)^2$. Given a constant $C_0\geq1$ and a matrix $\bSigma\in\calK$, define the following set of source sequences:
\begin{align}
\label{eq:shell}
\calT_n(\bSigma)&:=\bigg\{\bx^n\in\bbR^{2\times n}:\norm{\hat{\bSigma}(\bx^n)-\bSigma}\leq \frac{C_0}{n}\bigg\}.
\end{align}
For each $q\in[N_n]$, define 
\begin{align}
\calT_{n,q}&:=\big\{\bx^n\in\bbR^{2\times n}:\hat{\bSigma}(\bx^n)\in\calK_0,\ \mathfrak{q}_n(\hat{\bSigma}(\bx^n))=\bSigma_q\big\}.
\end{align}
It follows from~\eqref{eq:quantizer-error} that $\calT_{n,q}\subseteq\calT_n(\bSigma_q)$. We remark that $\calT_n(\bSigma)$ is analogous to the typical set while $\calT_{n,q}$ is analogous to the type class. Recall the definition of the minimum common rate in~\eqref{eq:R0-yu}.

Fix non-negative real numbers $(r_0,r_1,r_2)\in\bbR_{+,0}^3$, define the index set
\begin{align}
\label{eq:coded-index-set}
\calJ_n&:=\{q\in[N_n]:~\rvR_0(r_1,r_2,D_1,D_2|Q_{XY}^{\bSigma_q})\leq r_0\}.
\end{align}
For each $q\in[N_n]$, define the following set of covariance matrices:
\begin{align}
\label{def:Fq}
\calF_q&:=\big\{\bSigma_{\rmN}\in\bbS_+^2:\bSigma_{\rmN}\preceq\bSigma_q,\ \Sigma_{\mathrm N,jj}\leq D_je^{2r_j}~\forall~j\in[2]\big\}.
\end{align}
Furthermore, let
\begin{align}
\label{eq:cell-residual-choice}
\bSigma_{\mathrm N,q}&:=\underset{\bSigma_{\rmN}\in\calF_q}{\operatorname{arg\,max}}\ |\bSigma_{\rmN}|
\end{align}
be the covariance matrix in $\calF_q$ that has maximum determinant.
Finally, when $n\geq3$, define 
\begin{align}
\label{eq:sn-definition}
s_n&:=\frac{\sqrt n}{\log n},
\end{align}
and give real number $\kappa\in(0,1)$, define a parameter
\begin{align}
\label{eq:private-radius-constant}
\chi_\kappa&:=2+2\sqrt{\kappa^{-1}+1}.
\end{align}
The definitions above will be used to present our coding scheme in the next subsection.

\subsection{Coding Scheme}
\label{sec:code-construction}
\subsubsection{Codebook Generation}
For each $q\in\calJ_n$, we generate one common codebook and two private codebooks. The encoder and both decoders know these codebooks and the quantized centers $\mathfrak{q}_n(\calK_0)=(\Sigma_1,\ldots,\Sigma_{N_n})$ before transmission. 

Fix $q\in\calJ_n$. Define the following constant
\begin{align}
\label{eq:M0-random}
M_{0,q}&:=\left\lceil\exp\left\{\frac n2\log\frac{|\bSigma_q|}{|\bSigma_{\mathrm N,q}|}+\frac{s_n}{8}\right\}\right\rceil.
\end{align}
For each $m_0\in[M_{0,q}]$, generate $\bu_q^n(m_0)\in\bbR^{2\times n}$ i.i.d. from the two-dimensional Gaussian distribution $\calN(\bzero_2,\bone_{2\times 2})$ columns and let
\begin{align}
\bw_q^n(m_0)&:=(\bSigma_q-\bSigma_{\mathrm N,q})^{1/2} \bu_q^n(m_0).
\end{align}
Note that each column of $\bw_q^n(m_0)$ has distribution $\calN(\bzero_2,\bSigma_q-\bSigma_{\mathrm N,q})$. For each $j\in[2]$, define the following constant:
\begin{align}
\xi_{j,q,n}&:=\sqrt{n[\bSigma_{\mathrm N,q}]_{jj}+\chi_\kappa s_n}.
\end{align}
The next lemma shows that there exists deterministic private codebooks that cover two source sequences within desired distortion levels, which is proved in Appendix~\ref{app:euclidean-uniform}.
\begin{lemma}
\label{lem:euclidean-uniform}
Given any $\kappa\in(0,1)$, there exists a constant $C_{\mathrm{cov}}\geq1$, depending only on $(\kappa,D_1,D_2)$, such that the following result holds. For each $n\geq3$, and $q\in\calJ_n$ satisfying $\textcolor{black}{\bSigma_{\mathrm N,q}}\preceq\kappa^{-1}\bone_2$, and each $j\in[2]$, there exists a deterministic private codebook $\{v_{j,q}^n(m_j):m_j\in[M_{j,q}]\}\subset\bbR^n$ such that
\begin{align}
\label{eq:private-codebook-cover}
M_{j,q}&\leq C_{\mathrm{cov}}n^4\max\left\{1,\left(\frac{\xi_{j,q,n}^2}{nD_j}\right)^{n/2}\right\},\\
\calB_n(\bzero_n,\xi_{j,q,n})&\subseteq\bigcup_{m_j\in[M_{j,q}]}\calB_n(v_{j,q}^n(m_j),\sqrt{nD_j}).
\end{align}
\end{lemma}

\subsubsection{Encoding and Decoding}
We now specify the encoders and decoders using codebooks $\{\bw_q^n(m_0):~m_0\in[M_{0,q}],~q\in\calJ_n\}$, $\{v_{1,q}^n(m_1):~m_1\in[M_{1,q}]\}$ and $\{v_{2,q}^n(m_2):~m_2\in[M_{2,q}]\}$. Recall that $\bx^n=(x^n,y^n)$, $(\bx^n)_1=x^n$ and $(\bx^n)_2=y^n$. For each $q\in\calJ_n$, $m_0\in[M_{0,q}]$, and $j\in[2]$, define the following sets
\begin{align}
\calA_{0,n,q}(\bx^n)&:=\Big\{m_0\in[M_{0,q}]:\norm{(\bx^n)_j-\bw_q^n(m_0)_j}^2\leq\xi_{j,q,n}^2,~\forall~j\in[2]\Big\},
\\
\calA_{j,n,q}(\bx^n,m_0)&:=\Big\{m_j\in[M_{j,q}]:\norm{(\bx^n)_j-\bw_q^n(m_0)_j-v_{j,q}^n(m_j)}^2\leq nD_j\Big\}.
\end{align}
The set $\calA_{0,n,q}(\bx^n)$ collects the indices of common codewords such that the distortion between each source sequence and the common codeword is within distortion level that can be covered by private codewords in Lemma \ref{lem:euclidean-uniform} while the set $\calA_{j,n,q}(\bx^n,m_0)$ collects indices of private codewords that cover the $j$-th source sequence within target distortion levels.

Let $\calH_{n,q}:=\{\bx^n:\calA_{0,n,q}(\bx^n)\ne\emptyset\}$. Let $0$ be the fallback common symbol and define the following alphabet:
\begin{align}
\calM_{0,n}&:=\{0\}\cup\Big(\cup_{q\in\calJ_n}(\{q\}\times[M_{0,q}])\Big),
\end{align}
whose size is
\begin{align}
\label{eq:common-global-size}
M_0&:=1+\sum_{q\in\calJ_n}M_{0,q}.
\end{align}
Furthermore, for each $j\in[2]$, define
\begin{equation}
\label{eq:private-global-size}
M_j:=
\begin{cases}
\displaystyle\max_{q\in\calJ_n}M_{j,q}, & \text{if }\calJ_n\neq\varnothing,\\
1, & \text{if }\calJ_n=\varnothing.
\end{cases}
\end{equation}

Let $\iota_n:\calM_{0,n}\to[M_0]$ be a fixed bijection with $\iota_n(0)=1$, known to the encoder and both decoders. This function enables us to assign a single common message for each pair $(q,m_0)\in\calJ_n\times[M_{0,q}]$. For each $q\in\calJ_n$, define an encoder $f_q=(f_{0,q},f_{1,q},f_{2,q})$ such that if $\bx^n\in\calH_{n,q}$,  $f_{0,q}(\bx^n)$ is the smallest index in $\calA_{0,n,q}(\bx^n)$ and, for each $j\in[2]$, $f_{j,q}(\bx^n)$ is the smallest index in $\calA_{j,n,q}(\bx^n,f_{0,q}(\bx^n))$; otherwise, $f_q(\bx^n)=(1,1,1)$. Now, define the encoder $f=(f_0,f_1,f_2)$ as follows:
\begin{align}
\label{def:ourencoder}
f(\bx^n)&:=
\begin{cases}
\big(\iota_n(q,f_{0,q}(\bx^n)),f_{1,q}(\bx^n),f_{2,q}(\bx^n)\big),&\text{if }\bx^n\in\calT_{n,q}\cap\calH_{n,q},\mathrm{~for~some~} q\in\calJ_n,\\
\big(\iota_n(0),1,1\big),&\text{otherwise}
\end{cases}.
\end{align}
We next specify the decoder operation. Fix any $j\in[2]$. For each $q\in\calJ_n$, define
\begin{align}
\label{eq:shell-decoder-rule}
\phi_{j,q}(m_0,m_j)&:=\bw_q^n(m_0)_j+v_{j,q}^n(m_j).
\end{align}
Given a common message other than $\iota_n(0)$, decoder $\phi_j$ uses the inverse function $\iota_n^{-1}$ to obtain $(q,m_0)$ use compresses messages $(m_0,m_j)$ to generate the source sequence estimate as follows:
\begin{align}
\phi_j(\iota_n(q,m_0),m_j)&:=\phi_{j,q}(m_0,m_j).
\end{align}
If the common message equals $\iota_n(0)$, each decoder outputs source sequence estimate as $\bzero_n$.

\subsection{Covariance Covering}
\label{sec:assembly}

Recall from \eqref{def:quansetK0} that  $\mathfrak{q}_n(\calK_0)=(\Sigma_1,\ldots,\Sigma_{N_n})$ and the definitions of $\calF_q$ in \eqref{def:Fq} and the optimizer $\bSigma_{\mathrm N,q}$ in~\eqref{eq:cell-residual-choice}. The following lemma bounds the message sizes required to ensure that the reproduction of our coding scheme satisfies the distortion constraints for both decoders, which named covariance shell covering lemma.
\begin{lemma}
\label{lem:shell-cover}
There exist constants $\kappa\in(0,1)$ and $C\in\bbR_+$, and an integer $n_0\in\bbN$ such that for every $n\geq n_0$ and each $q\in\calJ_n$, $\bx^n\in\calT_n(\bSigma_q)$ and $j\in[2]$,
\begin{align}
\label{eq:shell-distortion-guarantee}
d\big((\bx^n)_j,\phi_{j,q}(f_{0,q}(\bx^n),f_{j,q}(\bx^n))\big)&\leq D_j,\\
\label{eq:shell-M0}
\log M_{0,q}&\leq\frac n2\log\frac{|\bSigma_q|}{|\bSigma_{\mathrm N,q}|}+Cs_n,\\
\label{eq:shell-Mj}
\log M_{j,q}&\leq\frac n2\left[\log\frac{[\bSigma_{\mathrm N,q}]_{jj}}{D_j}\right]^++Cs_n+C\log n.
\end{align}
\end{lemma}
The proof of Lemma~\ref{lem:shell-cover} is provided in Appendix~\ref{app:shell-cover}.

Fix three non-negative real numbers $(r_0,r_1,r_2)\in\bbR_{+,0}^3$. Recall the definitions of the joint excess-distortion probability $\veps_n$ in~\eqref{eq:Pe}, the empirical covariance $\Sigma(\cdot)$ in \eqref{eq:emp-cov} and $s_n$ in~\eqref{eq:sn-definition}. Let $C$ be the constant in Lemma \ref{lem:shell-cover} and define $h_n:=C(s_n+\log n)$. For each for $i\in[n]$, let $\bX_i:=(X_i,Y_i)^{\rvT}$ and let $\bX^n:=[\bX_1,\ldots,\bX_n]$, and $\hat{\bSigma}_n:=\hat{\bSigma}(\bX^n)$. The following lemma upper bounds the excess-distortion probability of our coding scheme.
\begin{lemma}
\label{lem:operational}
With the constant $C$ in $h_n$ chosen sufficiently large, there exists an integer $n_0$ such that the following holds for every $n\geq n_0$, uniformly over $(r_0,r_1,r_2)\in\bbR_{+,0}^3$. Both $C$ and $n_0$ are independent of the design rates. There exists an $(n,M_0,M_1,M_2)$-code per Def. \ref{def:code} such that for each $j\in[0:2]$,
\begin{align}
\label{eq:assembled-sizes}
\log M_j&\leq nr_j+h_n,
\end{align}
and 
\begin{align}
\label{eq:operational-bound}
\veps_n&\leq P_{XY}^n\big\{\hat{\bSigma}_n\notin\calK_0\big\}
+P_{XY}^n\big\{\hat{\bSigma}_n\in\calK_0,\ \rvR_0(r_1,r_2,D_1,D_2|Q_{XY}^{\mathfrak{q}_n(\hat{\bSigma}_n)})>r_0\big\}.
\end{align}
\end{lemma}
\begin{proof}
We first bound the excess-distortion probability.
It follows from our coding scheme and Lemma \ref{lem:shell-cover} that the joint excess-distortion probability satisfies
\begin{align}
\veps_n&\leq P_{XY}^n\{\hat{\bSigma}_n\notin\calK_0\}
+\sum_{q\in[N_n]\setminus\calJ_n}P_{XY}^n\{\hat{\bSigma}_n\in\calK_0,\mathfrak{q}_n(\hat{\bSigma}_n)=\bSigma_q\},\label{eq:pf-assembly-error-cells}\\
&=P_{XY}^n\{\hat{\bSigma}_n\notin\calK_0\}
+P_{XY}^n\{\hat{\bSigma}_n\in\calK_0,\rvR_0(r_1,r_2,D_1,D_2|Q_{XY}^{\mathfrak{q}_n(\hat{\bSigma}_n)})>r_0\},\label{eq:pf-assembly-error-event}
\end{align}
where \eqref{eq:pf-assembly-error-event} follows from the definition of $\calJ_n$ in \eqref{eq:coded-index-set}, which implies that if $q\notin\calJ_n$, $\rvR_0(r_1,r_2,D_1,D_2|Q_{XY}^{\Sigma_q})>r_0$.

We next show that the rate constraints are satisfied. For each $q\in\calJ_n$, it follows from the Gaussian first-order characterization
in~\cite[Theorem~3 and Corollary~2]{yu2023} and the covariance representation verified in Appendix~\ref{app:shell-cover} that
\begin{align}
\rvR_0(r_1,r_2,D_1,D_2|Q_{XY}^{\bSigma_q})
&=\min_{\substack{\bSigma_{\rmN}\in\calF_q\\
\bSigma_{\rmN}\succ0}}
\frac12\log\frac{|\bSigma_q|}{|\bSigma_{\rmN}|}\nonumber\\
&=\frac12\log\frac{|\bSigma_q|}{|\bSigma_{\mathrm N,q}|}
\leq r_0,\label{eq:cell-common-budget}
\end{align}
where the second equality follows from~\eqref{eq:cell-residual-choice},
and the inequality follows since $q\in\calJ_n$ (cf. \eqref{eq:coded-index-set}). Furthermore, it follows from the definition of the optimizer $\bSigma_{\mathrm N,q}$ in~\eqref{eq:cell-residual-choice}, in particular, the constraint $\bSigma_{\mathrm N,q}\in\calF_q$ (cf. \eqref{def:Fq}) that for each $j\in[2]$,
\begin{align}
\label{eq:cell-private-budget}
\frac12\left[\log\frac{[\bSigma_{\mathrm N,q}]_{jj}}{D_j}\right]^+&\leq r_j.
\end{align}
Recall from our coding scheme that if $\calJ_n=\varnothing$, we have $M_0=M_1=M_2=1$, and rate constraints are satisfied. In the following, we consider the case $\calJ_n\neq\varnothing$. It follows from Lemma~\ref{lem:shell-cover} that
\begin{align}
M_0&\leq1+N_ne^{nr_0+Cs_n}\leq \exp\bigl(nr_0+Cs_n+C\log n\bigr),
\end{align}
where the first inequality follows from the results in~\eqref{eq:shell-M0} and \eqref{eq:cell-common-budget}, and the second inequality follows from \eqref{eq:quantized-center-count} which implies that $N_n\leq Cn^3$. Furthermore, for each $j\in[2]$, it follows from \eqref{eq:cell-private-budget} and~\eqref{eq:shell-Mj} that
\begin{align}
M_j&\leq \exp\bigl(nr_j+Cs_n+C\log n\bigr).
\end{align}
Choosing the uniform constant $C$ in $h_n$ large enough to absorb the cell-index term verifies the rate constraints. 
\end{proof}

\subsection{Final Steps}
\label{sec:smooth-direct}
Fix a boundary rate tuple $\br^*=(R_0^*,R_1^*,R_2^*)$ with $R_0^*>0$, and fix $(L_0,L_1,L_2)$ satisfying the constraints on the right-hand side of~\eqref{eq:smooth-region}. Define $b_n:=(\log n)^{-1/4}$ and 
\begin{align}
\label{eq:target-common}
r_{0,n}&:=R_0^*+\frac{L_0+b_n}{\sqrt n},\\
\label{eq:rjn-design}
r_{j,n}&:=R_j^*+\frac{L_j}{\sqrt n},~j\in[2].
\end{align}
The design rates are nonnegative for all sufficiently large $n$. Indeed, if $R_j^*>0$,
$r_{j,n}>0$ eventually; if $R_j^*=0$, the constraint $L_j\geq0$ in
Theorem~\ref{thm:main} gives $r_{j,n}\geq0$. For all sufficiently large $n$
and each $j\in[0:2]$, Lemma~\ref{lem:operational} with
$(r_0,r_1,r_2)=(r_{0,n},r_{1,n},r_{2,n})$ implies
\begin{align}
\frac{\log M_{j,n}-nR_j^*}{\sqrt n}&\leq L_j+b_n+\frac{h_n}{\sqrt n}.
\end{align}
Since $b_n\to0$ and $h_n=o(\sqrt n)$, it follows that for each $j\in[0:2]$,
\begin{align}
\label{eq:overhead-vanishes}
\limsup_{n\to\infty}\frac{\log M_{j,n}-nR_j^*}{\sqrt n}&\leq L_j.
\end{align}
Furthermore, substituting~\eqref{eq:target-common} and~\eqref{eq:rjn-design} into~\eqref{eq:operational-bound} leads to
\begin{align}
\label{eq:direct-start}
\veps_n&\leq P_{XY}^n\{\hat{\bSigma}_n\notin\calK_0\}+P_{XY}^n\left\{\hat{\bSigma}_n\in\calK_0,
\sqrt n[\rvR_0(r_{1,n},r_{2,n},D_1,D_2|Q_{XY}^{\mathfrak{q}_n(\hat{\bSigma}_n)})-R_0^*]>L_0+b_n\right\}.
\end{align}
In the remaining proof, it suffices to show that $\limsup_{n\to\infty}\veps_n\leq\veps$.

We now bound the first term in \eqref{eq:direct-start}. Recall that $\hat{\bSigma}_n=\hat{\bSigma}(\bX^n)$. Since the source is Gaussian and memoryless,
\begin{equation}
\bbE_{P_{XY}^n}\Big[\norm{\textcolor{black}{\hat{\bSigma}_n}-\textcolor{black}{\bSigma_{XY}}}^2\Big]
=\frac{(\tr\textcolor{black}{\bSigma_{XY}})^2+\tr(\textcolor{black}{\bSigma_{XY}^2})}{n}.
\end{equation}
Furthermore, it follows from the definition of $\calK_0$ in~\eqref{eq:inner-covariance-neighborhood} and Markov's inequality that
\begin{align}
\label{eq:atypical-covariance-vanishes}
P_{XY}^n\{\textcolor{black}{\hat{\bSigma}_n}\notin\calK_0\}
&\leq\frac{(\tr\textcolor{black}{\bSigma_{XY}})^2+\tr(\textcolor{black}{\bSigma_{XY}^2})}{n\eta_{\mathrm{cov}}^2}\\
&=\frac{2\sigma_X^4+2\sigma_Y^4+2\sigma_X^2\sigma_Y^2(1+\rho_{XY}^2)}{n\eta_{\mathrm{cov}}^2}
=O\bigg(\frac{1}{n}\bigg).
\end{align}

We next bound the second term of~\eqref{eq:direct-start}. Recall that
$\bX=(X,Y)^{\rvT}$. For $\bSigma\in\calK$, the definition
of the distortion-tilted information density in~\eqref{eq:Gaussian-influence-parallel}
gives
\begin{align}
\bbE_{Q_{XY}^{\bSigma}}[\jmath(X,Y)]
&=R_0^*+\bbE_{Q_{XY}^{\bSigma}}[\tr(\bg\bX\bX^{\rvT})]-\tr(\bg\bSigma_{XY})\nonumber\\
&=R_0^*+\tr(\bg\bSigma)-\tr(\bg\bSigma_{XY})\nonumber\\
&=R_0^*+\ip{\bg}{\bSigma-\bSigma_{XY}},
\end{align}
where we used $\bX^{\rvT}\bg\bX=\tr(\bg\bX\bX^{\rvT})$ and
$\bbE_{Q_{XY}^{\bSigma}}[\bX\bX^{\rvT}]=\bSigma$. By
Proposition~\ref{ass:regularity}, as $(\bSigma,r_1,r_2)$ approaches
$(\bSigma_{XY},R_1^*,R_2^*)$ through $\calK\times\bbR_{+,0}^2$,
\begin{align}
\label{eq:local-differential}
\rvR_0(r_1,r_2,D_1,D_2|Q_{XY}^{\bSigma})
&=R_0^*+\ip{\bg}{\bSigma-\bSigma_{XY}}
-\sum_{j\in[2]}\lambda_j^*(r_j-R_j^*)\nonumber\\
&\quad+o\Bigl(\norm{\bSigma-\bSigma_{XY}}+\sum_{j\in[2]}|r_j-R_j^*|\Bigr).
\end{align}
We will apply \eqref{eq:local-differential} to $\bSigma=\mathfrak{q}_n(\hat{\bSigma}_n)$. Define
\begin{align}
\rvA_n
&:=\sqrt n\ip{\bg}{\hat{\bSigma}_n-\bSigma_{XY}},
\\
&=\frac1{\sqrt n}\sum_{i\in[n]}
\big[\bX_i^{\rvT}\bg\bX_i-\tr(\bg\bSigma_{XY})\big],\label{eq:ach-centered-quadratics}\\
&=\frac1{\sqrt n}\sum_{i\in[n]}
\big[\jmath(X_i,Y_i)-R_0^*\big],\label{eq:ach-quadratic-sum}
\end{align}
where Eq.~\eqref{eq:ach-centered-quadratics} follows from~\eqref{eq:emp-cov}, and Eq.~\eqref{eq:ach-quadratic-sum} follows from~\eqref{eq:Gaussian-influence-parallel}.

Define the third absolute moment by $\beta_3:=\bbE_{P_{XY}}[|\jmath(X,Y)-R_0^*|^3]$. Since $\jmath(X,Y)$ in \eqref{eq:Gaussian-influence-parallel} is a quadratic function of a Gaussian vector,
\begin{align}
\beta_3&\leq4\norm{\bg}^3\bbE_{P_{XY}}\big[\norm{\bX}^6\big]+4|\tr(\bg\bSigma_{XY})|^3<\infty,
\end{align}
where the last inequality follows since Gaussian random variables have finite sixth moments. Invoking the Berry--Esseen theorem~\cite[Theorem~1.1]{bentkus2005}, we conclude that there exists a finite positive universal constant $C_{\mathrm{BE}}$ such that
\begin{align}
\label{eq:ach-scalar-clt}
\sup_{t\in\bbR}\left|P_{XY}^n\{\rvA_n\leq t\}-\Phi(t/\sqrt V)\right|
&\leq\frac{C_{\mathrm{BE}}\beta_3}{V^{3/2}\sqrt n}.
\end{align}

For any $\hat{\bSigma}_n\in\calK_0$, define
\begin{align}
E_n
&:=\sqrt n\big[\rvR_0(r_{1,n},r_{2,n},D_1,D_2
\mid Q_{XY}^{\mathfrak{q}_n(\hat{\bSigma}_n)})-R_0^*\big]-\rvA_n+\lambda_1^*L_1+\lambda_2^*L_2.
\label{eq:global-Taylor-remainder}
\end{align}
When $\hat{\bSigma}_n\notin\calK_0$, define $E_n:=0$. Furthermore, let
$\ell:=L_0+\lambda_1^*L_1+\lambda_2^*L_2$. It follows that
\begin{align}
\sqrt n\big[\rvR_0(r_{1,n},r_{2,n},D_1,D_2
\mid Q_{XY}^{\mathfrak{q}_n(\hat{\bSigma}_n)})-R_0^*\big]
&=\rvA_n-\lambda_1^*L_1-\lambda_2^*L_2+E_n.
\end{align}
Fix $\zeta>0$. For sufficiently large $n$, it follows from~\eqref{eq:direct-start} that
\begin{align}
\veps_n&\leq P_{XY}^n\{\rvA_n+E_n>\ell+b_n\}
+P_{XY}^n\{\hat{\bSigma}_n\notin\calK_0\},\label{eq:direct-error-remainder}\\
&\leq P_{XY}^n\{\rvA_n>\ell-\zeta\}+P_{XY}^n\{|E_n|>\zeta\}
+P_{XY}^n\{\hat{\bSigma}_n\notin\calK_0\},\label{eq:direct-prob}\\
&\leq\rmQ\left(\frac{\ell-\zeta}{\sqrt V}\right)
+\frac{C_{\mathrm{BE}}\beta_3}{V^{3/2}\sqrt n}
+P_{XY}^n\{|E_n|>\zeta\}+P_{XY}^n\{\hat{\bSigma}_n\notin\calK_0\},\label{eq:direct-berry-bound}
\end{align}
where Eq.~\eqref{eq:direct-error-remainder} uses the exact remainder identity
on $\{\hat{\bSigma}_n\in\calK_0\}$, Eq.~\eqref{eq:direct-prob}
follows from the union bound and $b_n\geq0$, and
Eq.~\eqref{eq:direct-berry-bound} follows from~\eqref{eq:ach-scalar-clt}.
The last term in~\eqref{eq:direct-berry-bound} vanishes
by~\eqref{eq:atypical-covariance-vanishes}. Furthermore,
Appendix~\ref{app:taylor-remainder} proves that for every $\zeta>0$,
\begin{align}
\label{eq:En-prob}
P_{XY}^n\{|E_n|>\zeta\}&\longrightarrow0.
\end{align}
Combining the above analyses, we have
\begin{align}
\label{eq:smooth-ach-fixed-zeta}
\limsup_{n\to\infty}\veps_n&\leq\rmQ\left(\frac{\ell-\zeta}{\sqrt V}\right).
\end{align}
Finally, letting $\zeta\to 0$ and using continuity of $\rmQ$ gives
\begin{align}
\limsup_{n\to\infty}\veps_n&\leq\rmQ\left(\frac{L_0+\lambda_1^*L_1+\lambda_2^*L_2}{\sqrt V}\right)\label{ach:finalstep}.
\end{align}
The achievability proof of Theorem \ref{thm:main} is completed by using the requirement that $\limsup\varepsilon\leq \varepsilon$ in \eqref{eq:second-error} 
and the result in \eqref{ach:finalstep}.

\section{Converse Proof}
\label{sec:converse}
This section presents the converse proof of Theorem~\ref{thm:main}. When $\lambda_1^*>0$ and $\lambda_2^*>0$, we use the converse bound of~\cite[Theorem~17]{liu2020} and specify how the entropy support formula of \cite{liu2020} leads to the desired converse bound. When either $\lambda_1^*=0$ or $\lambda_2^*=0$, we show how the problem reduces to the P2P case and uses the converse bound~\cite[Theorem~40]{kostinaverdu2012}.

Fix any $\gamma>0$, boundary rate tuple $\br^*=(R_0^*,R_1^*,R_2^*)$ and second-order rates $(L_0,L_1,L_2)$ in Theorem \ref{thm:main}. For all sufficiently large $n$, we consider any $(n,M_{0,n},M_{1,n},M_{2,n})$-code per Def. \ref{def:code} such that
\begin{equation}\label{eq:rate-envelope-converse}
\log M_{j,n}\leq nR_j^*+\sqrt n(L_j+\gamma),\qquad j\in[0:2].
\end{equation}

\subsection{Positive Supporting Weights}
\label{sec:smooth-converse}
We first consider the case $\blambda^*=(\lambda_1^*,\lambda_2^*)\in(0,1]^2$. Given a covariance matrix $\bSigma\in\calK$ (cf. \eqref{eq:covariance-neighborhoods}), it follows from the determinant characterization in Appendix~\ref{app:shell-cover} and~\eqref{eq:pf-explicit-rate-optimal} that for any $(r_1,r_2)\in\bbR_{+,0}^2$,
\begin{equation}\label{eq:converse-residual-minimum}
\rvR_0(r_1,r_2,D_1,D_2\mid Q_{XY}^{\bSigma})
=\inf_{\substack{\bzero_{2\times2}\prec\bSigma_{\rmN}\preceq\bSigma\\
[\bSigma_{\rmN}]_{jj}\leq D_je^{2r_j},\ j\in[2]}}
\frac12\log\frac{|\bSigma|}{|\bSigma_{\rmN}|}.
\end{equation}
Thus, we have
\begin{align}
\inf_{\br\in\calR(D_1,D_2\mid Q_{XY}^{\bSigma})}
(r_0+\lambda_1^*r_1+\lambda_2^*r_2)
&=\inf_{(r_1,r_2)\in\bbR_{+,0}^2}\big[\rvR_0(r_1,r_2,D_1,D_2\mid Q_{XY}^{\bSigma})+\lambda_1^*r_1+\lambda_2^*r_2\big]\\
&=\inf_{\bzero_{2\times2}\prec\bSigma_{\rmN}\preceq\bSigma}
\left\{\frac12\log\frac{|\bSigma|}{|\bSigma_{\rmN}|}
+\sum_{j\in[2]}\frac{\lambda_j^*}{2}\left[\log\frac{[\bSigma_{\rmN}]_{jj}}{D_j}\right]^+\right\}.
\label{eq:operational-entropy-support}
\end{align}

Recall the definition of the source covariance matrix $\bSigma_{XY}$ in \eqref{eq:source-cov}. Given a matrix $\bSigma\in\bbS_{++}^2$, define
\begin{align}\label{eq:entropy-support}
\underline F_{\blambda^*}(\bSigma)
&:=\inf_{\bzero_{2\times2}\prec\bSigma_{\rmN}\preceq\bSigma}
\left\{\frac12\log\frac{|\bSigma|}{|\bSigma_{\rmN}|}
+\sum_{j\in[2]}\frac{\lambda_j^*}{2}\log\frac{[\bSigma_{\rmN}]_{jj}}{D_j}\right\},\\
\label{eq:converse-covariance-functional}
\varphi_{\blambda^*,\bSigma_{XY}}(\bSigma)
&:=-\underline F_{\blambda^*}(\bSigma)+\frac12\log|\bSigma|
-\frac12\tr(\bSigma_{XY}^{-1}\bSigma).
\end{align}
It follows from~\cite[Proposition~5, Eqs.~(64)--(74)]{liu2020} that $\varphi_{\blambda^*,\bSigma_{XY}}(\bSigma)$ is concave. In \cite{liu2020}, a matrix $\bg_{\mathrm{ent},\blambda^*}\in\bbS^2$ is called a supergradient at $\bSigma_{XY}$ if for all $\bSigma\in\bbS_{++}^2$,
\begin{equation}\label{eq:supergradient-definition}
\varphi_{\blambda^*,\bSigma_{XY}}(\bSigma)
\leq\varphi_{\blambda^*,\bSigma_{XY}}(\bSigma_{XY})+
\ip{\bg_{\mathrm{ent},\blambda^*}}{\bSigma-\bSigma_{XY}}.
\end{equation}
Fix any $\bx=(x,y)^{\rvT}\in\bbR^2$. Define
\begin{equation}\label{eq:converse-fluctuation-definition}
g_{\blambda^*}(\bx):=\bx^{\rvT}\bg_{\mathrm{ent},\blambda^*}\bx
-\tr(\bg_{\mathrm{ent},\blambda^*}\bSigma_{XY}).
\end{equation}
Recall the definition of $\bg$ in \eqref{def:bg}. The following lemma presents the properties of $\underline F_{\blambda^*}(\bSigma)$ and $g_{\blambda^*}(\bx)$.
\begin{lemma}\label{lem:entropy-support}
When $(\lambda_1^*,\lambda_2^*)$ are both positive,  the following claims hold.
\begin{enumerate}
\item The value of the entropy support function satisfies
\begin{equation}\label{eq:entropy-support-value}
\underline F_{\blambda^*}(\bSigma_{XY})=R_0^*+\lambda_1^*R_1^*+\lambda_2^*R_2^*.
\end{equation}
\item The entropy support is differentiable at the source covariance and 
\begin{equation}\label{eq:entropy-support-gradient}
\left.\nabla_{\bSigma}\underline F_{\blambda^*}(\bSigma)\right|_{\bSigma=\bSigma_{XY}}=\bg.
\end{equation}
\item The unique supergradient in~\eqref{eq:supergradient-definition} is $\bg_{\mathrm{ent},\blambda^*}=-\bg$, and 
\begin{align}
g_{\blambda^*}(\bx)&=-[\jmath(x,y)-R_0^*]\label{eq:converse-fluctuation-variance},\\
\var_{P_{XY}}[g_{\blambda^*}(\bX)]&=V.
\end{align}
\end{enumerate}
\end{lemma}
The proof of Lemma \ref{eq:entropy-support} is provided in Appendix~\ref{app:entropy-identification}.

We now apply the positive-weight Gaussian converse of bound Liu et al.~\cite[Theorem~17]{liu2020}. Using the notation of Liu et al., take $m=2$, $Q_{Y^m}=P_{XY}$, $c_0=1$, $c_j=\lambda_j^*$ for each $j\in[2]$, and Lebesgue reference measures. It follows from~\cite[Theorem~17]{liu2020} that 
\begin{equation}\label{eq:liu-entropy-value}
d^*:=\sup_{P_{U\mid XY}}\{-I(X,Y;U)-\lambda_1^*h(X\mid U)-\lambda_2^*h(Y\mid U)\},
\end{equation}
where the supremum is over the conditional distribution $P_{U\mid XY}$ for which the mutual information $I(X,Y;U)$ and conditional entropies $(h(X\mid U),h(Y\mid U))$ are finite. It follows from the Gaussian optimality argument in~\cite[Proposition~5, Proposition~16, and Eq.~(190)]{liu2020} and~\eqref{eq:entropy-support} that
\begin{equation}\label{eq:liu-support-normalization}
d^*+\sum_{j\in[2]}\frac{\lambda_j^*}{2}\log(2\pi eD_j)
=-\underline F_{\blambda^*}(\bSigma_{XY}).
\end{equation}
Indeed, for a Gaussian common auxiliary random variable with conditional covariance $\bSigma_{\rmN}$, the conditional marginal entropies are $\tfrac12\log(2\pi e[\bSigma_{\rmN}]_{jj})$ for each $j\in[2]$. Substituting these entropies into~\eqref{eq:liu-entropy-value} gives~\eqref{eq:liu-support-normalization}. 

For each $j\in[2]$, define
\begin{equation}
v_{j,n}:=\operatorname{vol}\bigl(\calB_n(\bzero_n,\sqrt{nD_j})\bigr),
\end{equation}
where $\operatorname{vol}$ denotes Lebesgue volume. It follows from~\cite[Theorem~17, Eq.~(193)]{liu2020} that for any $\theta\in\bbR$, the joint excess-distortion probability is lower bounded as follows:
\begin{align}
\veps_n
&\geq P_{XY}^n\left\{nd^*+\sum_{i\in[n]}g_{\blambda^*}(\bX_i)\leq\theta\right\}-\exp\left\{\theta+\log M_{0,n}+\sum_{j\in[2]}\lambda_j^*\log M_{j,n}
+\sum_{j\in[2]}\lambda_j^*\log v_{j,n}\right\}.
\label{eq:liu-raw-threshold}
\end{align}
It follows from the ball-volume formula and Stirling's expansion (cf.~\cite[Eqs.~(195)--(196)]{liu2020}) that
\begin{align}
\log v_{j,n}
&=\log\frac{\pi^{n/2}(nD_j)^{n/2}}{\Gamma(n/2+1)}\\
&=\frac n2\log(2\pi eD_j)-\frac12\log(\pi n)+O(n^{-1}).
\label{eq:Gaussian-ball-volume}
\end{align}
Given any real number $a\in\bbR$, taking $\theta=nd^*-\sqrt n a$ in~\eqref{eq:liu-raw-threshold} and using~\eqref{eq:liu-support-normalization}--\eqref{eq:Gaussian-ball-volume}, we have that for sufficiently large $n$,
\begin{align}
\veps_n
&\geq P_{XY}^n\left\{\frac1{\sqrt n}\sum_{i\in[n]}g_{\blambda^*}(\bX_i)\leq-a\right\}-\exp\left\{\log M_{0,n}+\sum_{j\in[2]}\lambda_j^*\log M_{j,n}
-n\underline F_{\blambda^*}(\bSigma_{XY})-\sqrt n a+C\log n\right\},
\label{eq:liu-finite-threshold}
\end{align}
where $C$ is finite and depends only on the fixed distortions and weights. It follows from claim (iii) of Lemma~\ref{lem:entropy-support}, the sum $\sum_{i\in[n]}g_{\blambda^*}(\bX_i)$ equals $\sum_{i\in[n]}(R_0^*-\jmath(X_i,Y_i))$, which has variance $V$ and third absolute moment $\beta_3$ in~\eqref{eq:ach-scalar-clt}
Analogous to \eqref{eq:ach-scalar-clt}, it follows from the Berry--Esseen theorem that
\begin{equation}\label{eq:liu-berry-esseen}
P_{XY}^n\left\{\frac1{\sqrt n}\sum_{i\in[n]}g_{\blambda^*}(\bX_i)\leq-a\right\}
\geq\rmQ\Big(\frac a{\sqrt V}\Big)-\frac{C_{\mathrm{BE}}\beta_3}{V^{3/2}\sqrt n}.
\end{equation}

Recall that $\ell=L_0+\lambda_1^*L_1+\lambda_2^*L_2$. Fix $\gamma>0$ and set  $a_\gamma:=\ell+(2+\lambda_1^*+\lambda_2^*)\gamma$. It follows from \eqref{eq:rate-envelope-converse} and claim (i) of Lemma~\ref{lem:entropy-support} that for sufficiently large $n$,
\begin{align}
\log M_{0,n}+\sum_{j\in[2]}\lambda_j^*\log M_{j,n}
-n\underline F_{\blambda^*}(\bSigma_{XY})-\sqrt n a_\gamma
&\leq\sqrt n\big[\ell+(1+\lambda_1^*+\lambda_2^*)\gamma-a_\gamma\big]\\
&=-\gamma\sqrt n.
\label{eq:liu-envelope-remainder}
\end{align}
Combining~\eqref{eq:liu-finite-threshold}--\eqref{eq:liu-envelope-remainder} leads to
\begin{equation}\label{eq:liu-quantitative-specialization}
\veps_n\geq\rmQ\Big(\frac a{\sqrt V}\Big)
-\frac{C_{\mathrm{BE}}\beta_3}{V^{3/2}\sqrt n}
-\exp\{-\gamma\sqrt n+C\log n\}.
\end{equation}
Letting $\gamma\to 0$, we have
\begin{equation}\label{eq:liu-specialized-converse}
\liminf_{n\to\infty}\veps_n\geq\rmQ\left(\frac{L_0+\lambda_1^*L_1+\lambda_2^*L_2}{\sqrt V}\right).
\end{equation}
The converse proof of Theorem \ref{thm:main} with positive $(\lambda_1^*,\lambda_2^*)$ is complete by using the requirement that $\limsup\varepsilon\leq \varepsilon$ in \eqref{eq:second-error} and the result in \eqref{eq:liu-specialized-converse}, which leads to
\begin{equation}\label{eq:positive-weight-second-order}
L_0+\lambda_1^*L_1+\lambda_2^*L_2\geq\sqrt V\,\rmQ^{-1}(\veps).
\end{equation}

\subsection{A Zero Supporting Weight}
It follows from Lemma~\ref{lem:boundary-classification} and Appendix~\ref{app:regularity} that a zero supporting weight has either $(\lambda_1^*,\lambda_2^*)=(1,0)$ or $(0,1)$. In the first case,
\begin{equation}
R_0^*+R_1^*=R_X(P_X,D_1),~R_2^*=0,~\qquad V=1/2.
\end{equation}
In this case, we can simply ignore the second decoder and treat the encoder $f_0,f_1$ as a single encoder and the decoder $\phi_1$ as a single decoder. Such degradation leads to a P2P setting. Fix $0<\delta<1-\veps$. It follows from the P2P Gaussian converse bound in~\cite[Theorem~40]{kostinaverdu2012} that
\begin{equation}\label{eq:aggregate-scalar-converse}
\log(M_{0,n}M_{1,n})\geq nR_X(P_X,D_1)
+\sqrt{n/2}\,\rmQ^{-1}(\veps+\delta)+O(\log n).
\end{equation}
Furthermore, it follows from~\eqref{eq:rate-envelope-converse} that
\begin{align}
\frac{\log(M_{0,n}M_{1,n})-nR_X(P_X,D_1)}{\sqrt n}
&=\frac{\log M_{0,n}-nR_0^*+\log M_{1,n}-nR_1^*}{\sqrt n}\label{eq:aggregate-displacement}\\
&\leq L_0+L_1+2\gamma.
\label{eq:pf-aggregate-envelope}
\end{align}
Combining~\eqref{eq:aggregate-scalar-converse}--\eqref{eq:pf-aggregate-envelope}, letting $n\to\infty$, and then letting $\gamma\to 0$ and $\delta\to 0$ yields
\begin{equation}\label{eq:zero-weight-converse-final}
L_0+\lambda_1^*L_1+\lambda_2^*L_2
=L_0+L_1\geq\frac{\rmQ^{-1}(\veps)}{\sqrt2}
=\sqrt V\,\rmQ^{-1}(\veps).
\end{equation}
When $(\lambda_1^*,\lambda_2^*)=(0,1)$, we have the similar result. The converse proof of Theorem \ref{thm:main} with zero $\lambda_j^*$ with $j\in[2]$ is now completed.

\section{Conclusion}
\label{sec:conclusion}
In this paper, we derived exact second-order asymptotics for the lossy Gaussian Gray--Wyner problem at boundary rate tuples with a positive common rate under the quadratic distortion measure and a joint excess-distortion probability criterion. Our main contribution was an achievability bound that matches the existing converse bounds~\cite{liu2020,kostinaverdu2012}. We obtained this bound by deriving a covariance shell covering lemma analogous to the type covering lemma for discrete sources.

We next discuss future directions. Firstly, in this paper, we derived the second-order asymptotics. However, to accurately characterize the performance of optimal codes at small blocklengths, further refined asymptotics including the third-order asymptotics~\cite[Theorem~27]{kostinaverdu2012} or the refined moderate deviations~\cite{yavaskostinaeffros2024} are necessary.
Secondly, we considered boundary rate tuples with a positive common rate in this paper. It is possible that some boundary rate tuple has a zero common rate, which can still have a positive second-order rate. It would be worthwhile to study the corresponding second-order coding for the Gray-Wyner problem in this regime. Finally, one can generalize our second-order asymptotics  to other multiterminal lossy source coding problems concerning bivariate Gaussian sources, e.g., the Wyner--Ziv problem~\cite{wyner1978generalsources} or the Heegard--Berger problem~\cite{heegardberger1985}. The techniques in this paper, especially covariance shell covering could be helpful.

\section*{Declaration of AI Usage}
The authors formulated the Gaussian Gray-Wyner problem with the research gap and provided proof strategies for both achievability and converse parts. ChatGPT (Astra 6 Ultra mode) was used to test proof strategies, fill proof steps, prepare early drafts and check the correctness of proofs. The authors have reviewed AI-assisted derivations and assume full responsibility for the statements, proofs, and correctness of the mathematical content. In particular, the authors have rewritten each line of the main text and revised most contents in the appendix.

\appendices
\section{Proof of Lemma~\ref{lem:boundary-classification}}
\label{app:boundary-classification}
Recall from~\eqref{eq:alpha-star} that
\begin{equation*}
(\alpha^*,\beta^*)=\bigl(R_X(P_X,D_1)-R_1^*,R_Y(P_Y,D_2)-R_2^*\bigr).
\end{equation*}
Recall also $\nu_t=\sqrt{1-e^{-2t}}$ for $t\geq0$ from~\eqref{eq:nu-def}. A boundary tuple satisfies $R_1^*\leq R_X(P_X,D_1)$, $R_2^*\leq R_Y(P_Y,D_2)$, and $R_0^*=\Upsilon_{\rho_{XY}}(\alpha^*,\beta^*)$ by~\eqref{eq:R0-yu}. If $\alpha^*>0$ and $\beta^*=0$, then $R_2^*=R_Y(P_Y,D_2)>0$. Reducing $R_2^*$ by a small $t>0$ with $\nu_t<\rho_{XY}\nu_{\alpha^*}$ leaves the common rate equal to $\alpha^*$, contradicting Pareto minimality. Interchanging the sources excludes $\alpha^*=0$, $\beta^*>0$. If $\alpha^*=\beta^*=0$,~\eqref{eq:Upsilon-axes} implies $R_0^*=0$. Thus, $(\alpha^*,\beta^*)\in\bbR_+^2$.

On $\calP_2$, increasing $\alpha$ or $\beta$ while keeping the other fixed increases $\Upsilon_{\rho_{XY}}$ by the same amount. On $\calP_1$, we have $\nu_\alpha\geq\rho_{XY}\nu_\beta$ and $\nu_\beta\geq\rho_{XY}\nu_\alpha$ by~\eqref{eq:D1}. Substituting $1-\nu_t^2=e^{-2t}$ and~\eqref{eq:rho-ab} into~\eqref{eq:Upsilon} leads to
\begin{align*}
\Upsilon_{\rho_{XY}}(\alpha,\beta)
&=-\frac12\log\frac{(1-\nu_\alpha^2)(1-\nu_\beta^2)-(\rho_{XY}-\nu_\alpha\nu_\beta)^2}{1-\rho_{XY}^2}\\
&=\beta-\frac12\log\left(1-\frac{(\nu_\alpha-\rho_{XY}\nu_\beta)^2}{(1-\nu_\beta^2)(1-\rho_{XY}^2)}\right).
\end{align*}
Given $\beta>0$, increasing $\alpha$ strictly increases the nonnegative quantity $\nu_\alpha-\rho_{XY}\nu_\beta$ and its square, even when initially zero. Hence, $\Upsilon_{\rho_{XY}}$ strictly increases. If the pair leaves $\calP_1$, it enters $\calP_2$ or $\calP_3$, whose formulas agree at the boundary and increase with $\alpha$ at slope one. Interchanging $\alpha$ and $\beta$ proves the other strict increase. Thus, given $(\alpha^*,\beta^*)\in\calP_1\cup\calP_2$, reducing either positive private rate increases the required common rate. Reducing $R_0^*$ violates~\eqref{eq:first-order-epigraph}, and zero private rates cannot decrease, proving Pareto minimality.

The sum-rate formulas in (a) and (b) follow by substituting~\eqref{eq:Upsilon} into
\begin{equation*}
R_0^*+R_1^*+R_2^*=R_X(P_X,D_1)+R_Y(P_Y,D_2)+\Upsilon_{\rho_{XY}}(\alpha^*,\beta^*)-\alpha^*-\beta^*.
\end{equation*}

In the strict $\calP_3$ regime, $\Upsilon_{\rho_{XY}}(\alpha^*,\beta^*)=\alpha^*$. If $R_2^*>0$, a small reduction of $R_2^*$ stays in $\calP_3$ without changing the common rate, contradicting Pareto minimality. If $R_2^*=0$, reducing a positive $R_1^*$ preserves $\nu_{\beta^*}<\rho_{XY}\nu_{\alpha^*}$ and increases the required common rate by the same amount. Hence, the boundary condition is $R_2^*=0$ and $R_0^*+R_1^*=R_X(P_X,D_1)$. Interchanging the sources proves (d). The interfaces lie in $\calP_1\cup\calP_2$, completing the proof.

\section{Proof of Proposition~\ref{ass:regularity}}
\label{app:regularity}
\emph{Proof of Statement (i).} Given $(\alpha,\beta)\in\bbR_+^2$, let $r=r_{\alpha,\beta}$ and $u=(\nu_\beta/\nu_\alpha)e^{\beta-\alpha}$. Since $\nu_t'=e^{-2t}/\nu_t$, differentiation of~\eqref{eq:rho-ab} implies $\partial_\alpha r=r-u$, $\partial_\beta r=r-u^{-1}$, and $\partial_\rho r=e^{\alpha+\beta}$. Let $F_k$ denote expression $k$ of~\eqref{eq:Upsilon}. The curved expression has derivatives
\begin{align}
\partial_\alpha F_1
&=1+\frac{r(r-u)}{1-r^2}=\frac{1-ru}{1-r^2},\label{eq:pf-curved-slope}\\
\partial_\beta F_1
&=1+\frac{r(r-u^{-1})}{1-r^2}=\frac{1-r/u}{1-r^2},\label{eq:pf-curved-slope-beta}\\
\partial_\rho F_1
&=\frac{re^{\alpha+\beta}}{1-r^2}-\frac{\rho}{1-\rho^2}.\label{eq:r-rho}
\end{align}
On the other three regimes,
\begin{equation}\label{eq:branch-derivatives}
(\partial_\alpha F_k,\partial_\beta F_k,\partial_\rho F_k)=
\begin{cases}
(1,1,-\rho/(1-\rho^2)),&k=2,\\
(1,0,0),&k=3,\\
(0,1,0),&k=4.
\end{cases}
\end{equation}
The curved-regime constraints imply $0\leq r\leq\min\{u,u^{-1}\}<1$ except possibly when $u=1$. In that case $\alpha=\beta$, and $r=1-(1-\rho)e^{2\alpha}<1$.
Thus, $1-r^2>0$, and the derivatives of $F_1$ are continuous up to each boundary.

At the $(\calP_1,\calP_2)$ interface, $r=0$, so $F_1=F_2$ and~\eqref{eq:pf-curved-slope}--\eqref{eq:r-rho} imply $\nabla_{(\alpha,\beta,\rho)}F_1=(1,1,-\rho/(1-\rho^2))=\nabla_{(\alpha,\beta,\rho)}F_2$. At the $(\calP_1,\calP_3)$ interface, $\nu_\beta=\rho\nu_\alpha$, $r=u$, and $1-r^2=(1-\rho^2)e^{2\beta}$ imply $F_1=\alpha=F_3$ and $\nabla_{(\alpha,\beta,\rho)}F_1=(1,0,0)=\nabla_{(\alpha,\beta,\rho)}F_3$. Interchanging $\alpha$ and $\beta$ handles $(\calP_1,\calP_4)$. Therefore, the values and first derivatives agree across all boundaries, proving (i).

\emph{Proof of Statement (ii).} Lemma~\ref{lem:boundary-classification} implies $(\alpha^*,\beta^*)\in\bbR_+^2$. Hence, an open neighborhood $\calU$ of $(\bSigma_{XY},R_1^*,R_2^*)$ can be chosen where the covariance is positive definite, $\Sigma_{jj}>D_j$, $\Sigma_{12}>0$, and $\frac12\log(\Sigma_{jj}/D_j)-r_j>0$ for $j\in[2]$. 
Given $(\bSigma,r_1,r_2)\in\calU$, define
\begin{equation*}
\overline R_0(\bSigma,r_1,r_2)
:=\Upsilon_{\rho(\bSigma)}\left(\frac12\log\frac{\Sigma_{11}}{D_1}-r_1,\frac12\log\frac{\Sigma_{22}}{D_2}-r_2\right).
\end{equation*}
The correlation $\rho(\bSigma)=\Sigma_{12}/\sqrt{\Sigma_{11}\Sigma_{22}}$ and the logarithmic arguments have continuous first derivatives, so Statement (i) implies the same property for $\overline R_0$. Given $(r_1,r_2)\in\bbR_{+,0}^2$, Lemma~\ref{prop:first-order-selection} implies $\overline R_0(\bSigma,r_1,r_2)=\rvR_0(r_1,r_2,D_1,D_2\mid Q_{XY}^{\bSigma})$. Allowing small negative $r_j$ only defines the derivative at $r_j=0$.

Given a symmetric perturbation $\bh=[h_{jk}]$ and private-rate perturbations $(s_1,s_2)$, let $(\eta,\lambda_1,\lambda_2)$ denote $(\partial_\rho\Upsilon,\partial_\alpha\Upsilon,\partial_\beta\Upsilon)$ evaluated at $\bigl(\rho(\bSigma),\frac12\log(\Sigma_{11}/D_1)-r_1,\frac12\log(\Sigma_{22}/D_2)-r_2\bigr)$. Direct differentiation leads to
\begin{align}
\mathrm d\overline R_0[\bh,s_1,s_2]
&=\lambda_1\left(\frac{h_{11}}{2\Sigma_{11}}-s_1\right)
+\lambda_2\left(\frac{h_{22}}{2\Sigma_{22}}-s_2\right)\notag\\
&\quad+\eta\left(\frac{h_{12}}{\sqrt{\Sigma_{11}\Sigma_{22}}}
-\frac{\rho(\bSigma)h_{11}}{2\Sigma_{11}}
-\frac{\rho(\bSigma)h_{22}}{2\Sigma_{22}}\right)\notag\\
&=\frac{\lambda_1-\rho(\bSigma)\eta}{2\Sigma_{11}}h_{11}
+\frac{\eta}{\sqrt{\Sigma_{11}\Sigma_{22}}}h_{12}
+\frac{\lambda_2-\rho(\bSigma)\eta}{2\Sigma_{22}}h_{22}
-\lambda_1s_1-\lambda_2s_2.\label{eq:pf-covariance-inner-product}
\end{align}
At $(\bSigma,r_1,r_2)=(\bSigma_{XY},R_1^*,R_2^*)$, the identity $\ip{\bg}{\bh}=g_{11}h_{11}+2g_{12}h_{12}+g_{22}h_{22}$ leads to the matrix in~\eqref{eq:G-gradient} and the private-rate derivatives $-\lambda_j^*$ from~\eqref{eq:pf-covariance-inner-product}.

At $\nu_\beta=\rho\nu_\alpha$, the first $\beta$-derivative vanishes. A small increase of $\beta$ enters $\calP_1$, where differentiation of~\eqref{eq:pf-curved-slope-beta} leads to
\begin{equation*}
\partial_{\beta,+}^2\Upsilon_\rho(\alpha,\beta)
=u^{-2}+\frac{1}{\nu_\beta^2(1-u^2)}
=\frac{(1-\rho^2\nu_\alpha^2)^2}{\rho^2\nu_\alpha^2(1-\nu_\alpha^2)(1-\rho^2)}>0.
\end{equation*}
Indeed, $\partial_\beta r=r-u^{-1}$, $\partial_\beta\log u=\nu_\beta^{-2}$, and $u^2=\rho^2(1-\nu_\alpha^2)/(1-\rho^2\nu_\alpha^2)$ at this interface. Therefore, Taylor's formula leads to
\begin{equation}\label{eq:interface-quadratic}
\Upsilon_\rho(\alpha,\beta+h)=\alpha+
\frac{(1-\rho^2\nu_\alpha^2)^2}{2\rho^2\nu_\alpha^2(1-\nu_\alpha^2)(1-\rho^2)}h^2+o(h^2),\quad h\downarrow0.
\end{equation}
When $h<0$ is sufficiently small, the value remains $\alpha$ in $\calP_3$. The proof of Statement (ii) is completed.

\section{Proof of Lemma~\ref{lem:critical-quantities}}
\label{app:information-quantity}
Given $\bSigma\in\calK$ and $(X,Y)\sim Q_{XY}^{\bSigma}$, let $\bX=(X,Y)^{\rvT}$. By~\eqref{eq:Gaussian-influence-parallel} and $\bbE_{Q_{XY}^{\bSigma}}[\bX\bX^{\rvT}]=\bSigma$, we have $\bbE_{Q_{XY}^{\bSigma}}[\jmath(X,Y)]=R_0^*+\tr[\bg(\bSigma-\bSigma_{XY})]$. Taking $\bSigma=\bSigma_{XY}$ proves (i).

For the remainder of the proof, take $(X,Y)\sim P_{XY}$. Let $\bomega:=\bSigma_{XY}^{1/2}\bg\bSigma_{XY}^{1/2}$ have eigenvalues $\omega_1,\omega_2$. The coordinates $N_1,N_2$ of $\bSigma_{XY}^{-1/2}\bX\sim\calN(\bzero_2,\bone_2)$ along orthonormal eigenvectors of $\bomega$ are independent standard Gaussian variables. Therefore,
\begin{align}
\jmath(X,Y)-R_0^*&\overset{\mathrm d}=\omega_1(N_1^2-1)+\omega_2(N_2^2-1),\label{eq:pf-quadratic-law}\\
V&=2(\omega_1^2+\omega_2^2)=2\tr[(\bg\bSigma_{XY})^2].\label{eq:pf-quadratic-trace}
\end{align}
The variance identity uses $\var[N_j^2]=2$ and independence. The expression in~\eqref{eq:pf-quadratic-law} also implies $\bbE_{P_{XY}}[|\jmath(X,Y)-R_0^*|^3]<\infty$.

Let $\bSigma_{\rmN}^{*}$ be the optimizer in~\eqref{eq:cell-residual-choice} at $\bSigma_q=\bSigma_{XY}$ and $(r_1,r_2)=(R_1^*,R_2^*)$, and put $\tau_j=D_je^{2R_j^*}>0$. Equivalently, $\bSigma_{\rmN}^{*}$ minimizes $-\tfrac12\log|\bSigma_{\rmN}|$ under $\bSigma_{\rmN}\preceq\bSigma_{XY}$ and $[\bSigma_{\rmN}]_{jj}\leq\tau_j$. Choose $\delta>0$ sufficiently small that $\delta\bone_2\prec\bSigma_{XY}$ and $\delta<\tau_j$. Hence, the convex optimality conditions imply the existence of multipliers $\vartheta_j\geq0$ and $\blambda_{\mathrm{cov}}\succeq0$ satisfying
\begin{align}
-\frac12(\bSigma_{\rmN}^{*})^{-1}
+\diag(\vartheta_1,\vartheta_2)+\blambda_{\mathrm{cov}}
&=\bzero_{2\times2},\label{eq:spectrum-stationarity}\\
\ip{\blambda_{\mathrm{cov}}}{\bSigma_{XY}-\bSigma_{\rmN}^{*}}&=0,\label{eq:spectrum-complementarity}
\end{align}
together with $\vartheta_j([\bSigma_{\rmN}^{*}]_{jj}-\tau_j)=0$. 
Given any $\bSigma_{\rmN}\succ0$ with $\bSigma_{\rmN}\preceq\bSigma$ and $[\bSigma_{\rmN}]_{jj}\leq\tau_j$, convexity and~\eqref{eq:spectrum-stationarity} imply
\begin{align*}
-\frac12\log|\bSigma_{\rmN}|
&\geq-\frac12\log|\bSigma_{\rmN}^{*}|-\sum_{j\in[2]}\vartheta_j\bigl([\bSigma_{\rmN}]_{jj}-[\bSigma_{\rmN}^{*}]_{jj}\bigr)
-\ip{\blambda_{\mathrm{cov}}}{\bSigma_{\rmN}-\bSigma_{\rmN}^{*}}\\
&\geq-\frac12\log|\bSigma_{\rmN}^{*}|-\ip{\blambda_{\mathrm{cov}}}{\bSigma-\bSigma_{XY}}.
\end{align*}
The second inequality follows from the diagonal bounds, $\vartheta_j([\bSigma_{\rmN}^{*}]_{jj}-\tau_j)=0$, $\bSigma_{\rmN}\preceq\bSigma$, and~\eqref{eq:spectrum-complementarity}. Taking the infimum over the feasible $\bSigma_{\rmN}$ and using~\eqref{eq:converse-residual-minimum} yields
\begin{align}
\rvR_0(R_1^*,R_2^*,D_1,D_2\mid Q_{XY}^{\bSigma})-\frac12\log|\bSigma|
&\geq R_0^*-\frac12\log|\bSigma_{XY}|
-\ip{\blambda_{\mathrm{cov}}}{\bSigma-\bSigma_{XY}}.\label{eq:spectrum-dual-touch}
\end{align}
Equality holds in~\eqref{eq:spectrum-dual-touch} at $\bSigma=\bSigma_{XY}$. Differentiating there in both signs of every symmetric direction, using Proposition~\ref{prop:regularity}, yields
\begin{equation}\label{eq:spectrum-sensitivity}
\bg=\frac12\bSigma_{XY}^{-1}-\blambda_{\mathrm{cov}}.
\end{equation}

Set $\bSigma_{\mathrm W}^{*}:=\bSigma_{XY}-\bSigma_{\rmN}^{*}\succeq0$. Since $R_0^*>0$, we have $\bSigma_{\mathrm W}^{*}\neq\bzero_{2\times2}$. Positive semi-definiteness and~\eqref{eq:spectrum-complementarity} imply
\begin{align*}
0&=\tr(\blambda_{\mathrm{cov}}\bSigma_{\mathrm W}^{*})
=\norm{\blambda_{\mathrm{cov}}^{1/2}(\bSigma_{\mathrm W}^{*})^{1/2}}^2.
\end{align*}
Thus, $\blambda_{\mathrm{cov}}\bSigma_{\mathrm W}^{*}=\bzero_{2\times2}$ with $\bSigma_{\mathrm W}^{*}\neq\bzero_{2\times2}$, so $\rank(\blambda_{\mathrm{cov}})\leq1$. Equation~\eqref{eq:spectrum-sensitivity} implies
\begin{equation*}
\bomega=\frac12\bone_2-\bSigma_{XY}^{1/2}\blambda_{\mathrm{cov}}\bSigma_{XY}^{1/2},
\end{equation*}
so one eigenvalue of $\bomega$ is $1/2$. Since~\eqref{eq:G-gradient} implies $\tr(\bomega)=\tr(\bg\bSigma_{XY})=(\lambda_1^*+\lambda_2^*)/2$, the other is $(\lambda_1^*+\lambda_2^*-1)/2$. Substitution into~\eqref{eq:pf-quadratic-trace} proves~\eqref{eq:V-slopes}.

In $\calP_1$, let $r=r_{\alpha^*,\beta^*}$ and recall $u$ from~\eqref{eq:u-ratio}. The derivatives in~\eqref{eq:pf-curved-slope}--\eqref{eq:pf-curved-slope-beta} imply
\begin{align}
\lambda_1^*+\lambda_2^*-1
&=\frac{1+r^2-r(u+u^{-1})}{1-r^2}
=\frac{(u-r)(1-ru)}{u(1-r^2)}.\label{eq:spectrum-slope-range}
\end{align}
In the strict curved regime $0<r<\min\{u,u^{-1}\}$, so the last expression is positive, while each slope lies in $(0,1)$. In $\calP_2,\calP_3,\calP_4$ the slopes are respectively $(1,1),(1,0),(0,1)$ by~\eqref{eq:branch-derivatives}. Therefore, $1\leq\lambda_1^*+\lambda_2^*\leq2$, and substitution in~\eqref{eq:V-slopes} leads to all four expressions in~\eqref{eq:V-branches}, including their interfaces, together with $1/2\leq V\leq1$. Since $\tau_j>0$ and Proposition~\ref{prop:regularity} also applies when $R_j^*=0$, the proof includes zero private rates.

\section{Proof of Lemma~\ref{lem:euclidean-uniform}}
\label{app:euclidean-uniform}
Fix $T\geq1$, $n\geq3$, and $0\leq R\leq Tr$ with $r>0$. We construct a cover of $\calB_n(\bzero_n,R)$ by radius-$r$ balls with at most $C_Tn^4\max\{1,(R/r)^n\}$ centers.
If $R\leq r$, take $\calV=\{\bzero_n\}$. Suppose $R>r$, and put $h=r/(4n)$, $r_-=r-2h$, $m=\lceil(R-r)/h\rceil$, and $R_k=\min\{r+kh,R\}$ for $k\in[m]$.
On each sphere $\calS_k:=\{\bx\in\bbR^n:\norm\bx=R_k\}$, choose a maximal set $\calN_k$ with pairwise distances at least $h$. Maximality ensures that every point of $\calS_k$ is within distance $h$ of $\calN_k$. Packing radius-$h/2$ balls and comparing their volumes lead to
\begin{align}
|\calN_k|&\leq\left(1+\frac{2R_k}{h}\right)^n
\leq(1+8Tn)^n.\label{eq:ballcover-net}
\end{align}

Choose a center uniformly on the sphere of radius $a_k=\sqrt{R_k^2-r_-^2}$. Given $\bz\in\calN_k$, let $\theta$ be its angle with the center. The squared distance is $R_k^2+a_k^2-2R_ka_k\cos\theta$, so coverage within radius $r_-$ is equivalent to $\cos\theta\geq z_k:=a_k/R_k$. The first coordinate of a uniform point on the unit sphere has density $\omega_n(1-u^2)^{(n-3)/2}$, with $\omega_n:=\Gamma(n/2)/(\sqrt\pi\,\Gamma((n-1)/2))\geq1/2$. Substitution $w=1-u^2$ and $(1-w)^{-1/2}\geq1$ show that the coverage probability satisfies
\begin{align}
p_k&=\omega_n\int_{z_k}^1(1-u^2)^{(n-3)/2}\,du\notag\\
&=\frac{\omega_n}{2}\int_0^{r_-^2/R_k^2}w^{(n-3)/2}(1-w)^{-1/2}\,dw\notag\\
&\geq\frac{1}{2(n-1)}\left(\frac{r_-}{R_k}\right)^{n-1}.\label{eq:ballcover-cap-lower}
\end{align}
The inequality $\omega_n\geq1/2$ follows by integrating its density over $[-1,1]$, and the bound remains uniform when $R/r\downarrow1$.

Independently draw
\begin{equation}\label{eq:ballcover-random-count}
J_k:=\left\lceil2(n-1)\left(\frac{R_k}{r_-}\right)^{n-1}
[n\log(1+8Tn)+1]\right\rceil
\end{equation}
centers uniformly on the sphere of radius $a_k$. Since $J_kp_k\geq n\log(1+8Tn)+1$, the union bound and~\eqref{eq:ballcover-net} imply that the probability any point of $\calN_k$ is uncovered is at most
\begin{align*}
|\calN_k|(1-p_k)^{J_k}
&\leq(1+8Tn)^n e^{-J_kp_k}\leq e^{-1}<1.
\end{align*}
Fix sets $\calC_k$ of $J_k$ centers covering $\calN_k$, and put $\calV=\{\bzero_n\}\cup\bigcup_{k\in[m]}\calC_k$. Given $\bx\in\bbR^n$ with $r<\norm\bx\leq R$, choose $k$ and $\bz\in\calN_k$ such that $|\norm\bx-R_k|\leq h$ and $\norm{R_k\bx/\norm\bx-\bz}\leq h$. The triangle inequality implies
\begin{align*}
\min_{\bv\in\calV}\norm{\bx-\bv}
&\leq\norm{\bx-\frac{R_k}{\norm\bx}\bx}+\norm{\frac{R_k}{\norm\bx}\bx-\bz}+\min_{\bv\in\calC_k}\norm{\bz-\bv}\\
&\leq|\norm\bx-R_k|+h+r_-\leq2h+r_-=r.
\end{align*}
Given $\bx\in\bbR^n$ with $\norm\bx\leq r$, we have $\norm{\bx-\bzero_n}\leq r$.

Finally, $m\leq4nT+1$ and
\begin{align}
\left(\frac r{r_-}\right)^{n-1}
&=\left(1-\frac1{2n}\right)^{-(n-1)}\leq e^{1/2},\label{eq:ballcover-margin-cost}
\end{align}
where $-\log(1-x)\leq x/(1-x)$. Since $R_k\leq R$, the ceiling bound in~\eqref{eq:ballcover-random-count} yields
\begin{align}
|\calV|&\leq1+\sum_{k\in[m]}J_k\notag\\
&\leq1+m+2m(n-1)e^{1/2}(R/r)^{n-1}
[n\log(1+8Tn)+1]\notag\\
&\leq C_Tn^4(R/r)^n,\label{eq:ballcover-total-count}
\end{align}
where $C_T$ depends only on $T$, using $m\leq4nT+1$ and $1+\log n\leq n$ for $n\geq3$.

To apply the covering bound, take $T=\max\{1,\sqrt{(\kappa^{-1}+\chi_\kappa)/\min\{D_1,D_2\}}\}$. Since $s_n/n\leq1$ and $\bSigma_{\mathrm N,q}\preceq\kappa^{-1}\bone_2$, we have $\xi_{j,q,n}^2/(nD_j)\leq T^2$. Taking $R=\xi_{j,q,n}$ and $r=\sqrt{nD_j}$ leads to the codebooks in~\eqref{eq:private-codebook-cover}, with $C_{\mathrm{cov}}=C_T$ depending only on $(\kappa,D_1,D_2)$.

\section{Proof of Lemma~\ref{lem:shell-cover}}
\label{app:shell-cover}

Fix $\bSigma\in\calK$ and $(r_1,r_2)\in\bbR_{+,0}^2$. In the covariance optimization~\eqref{eq:cell-residual-choice}, replace the cell center $\bSigma_q$ by $\bSigma$ while retaining these private rates. Since $\norm{\bSigma-\bSigma_{XY}}\leq2\eta_{\mathrm{cov}}$, the eigenvalue and entry perturbation bounds in~\eqref{eq:explicit-covariance-radius} imply
\begin{align*}
\lambda_{\min}(\bSigma)&\geq\lambda_{\min}(\bSigma_{XY})-2\eta_{\mathrm{cov}}
\geq\frac34\lambda_{\min}(\bSigma_{XY})>0,\\
\Sigma_{jj}-D_j&\geq[\bSigma_{XY}]_{jj}-D_j-2\eta_{\mathrm{cov}}
\geq\frac34([\bSigma_{XY}]_{jj}-D_j)>0,\qquad j\in[2],
\end{align*}
where the first line follows from the Rayleigh quotient and $\norm{\ba}_{\mathrm{op}}\leq\norm{\ba}$ for a real matrix $\ba$, and the second uses the entrywise Frobenius bound. Compactness yields $0<m<M$ with $m\bone_2\preceq\bSigma\preceq M\bone_2$ throughout $\calK$. The covariance segment from $\bSigma_{XY}$ to $\bSigma$ has diagonal entries at least half their source values. Therefore, differentiating~\eqref{eq:rho-S} along this segment and applying the mean-value bound leads to
\begin{equation*}
|\rho(\bSigma)-\rho_{XY}|
\leq c_{\mathrm{cor}}\norm{\bSigma-\bSigma_{XY}}
\leq2c_{\mathrm{cor}}\eta_{\mathrm{cov}}
\leq\frac14\min\{\rho_{XY},1-\rho_{XY}\}.
\end{equation*}
Hence, $\rho_{XY}/2\leq\rho(\bSigma)\leq(1+\rho_{XY})/2$, so the Gaussian first-order formula applies uniformly on $\calK$.

Choose $0<\delta<\min\{m,D_1,D_2\}$ and $0<\kappa<1$ with $\kappa\leq\min\{\delta^2/M,1/M\}$. Given any nonnegative rate pair, $\delta\bone_2$ is feasible because $\delta\bone_2\preceq\bSigma$ and $\delta<D_j\leq D_je^{2r_j}$. Compactness of the feasible set and continuity of the determinant imply the existence of an optimizer $\bSigma_{\rmN}^{*}$ with $|\bSigma_{\rmN}^{*}|\geq\delta^2>0$. Consequently,
\begin{equation*}
\lambda_{\min}(\bSigma_{\rmN}^{*})
=\frac{|\bSigma_{\rmN}^{*}|}{\lambda_{\max}(\bSigma_{\rmN}^{*})}
\geq\frac{\delta^2}{M}\geq\kappa,
\end{equation*}
and $\bSigma\preceq M\bone_2\preceq\kappa^{-1}\bone_2$. Thus,
\begin{equation}
\label{eq:optimizer-compact}
\kappa\bone_2\preceq\bSigma_{\rmN}^{*}\preceq\bSigma\preceq\kappa^{-1}\bone_2.
\end{equation}
The same $\kappa$ applies to every $\bSigma\in\calK$ and nonnegative rate pair, including the optimizer $\bSigma_{\mathrm N,q}$ in $\calF_q$.

We next verify the optimizer's value explicitly. Let $\rho_{\bSigma}:=\rho(\bSigma)$, $\alpha:=\alpha(\bSigma;r_1,D_1)$, and $\beta:=\beta(\bSigma;r_2,D_2)$. Define $\bt_{\bSigma}:=\diag(\sqrt{\Sigma_{11}},\sqrt{\Sigma_{22}})$ and
\begin{equation*}
v_1:=e^{-2\alpha}=\min\{1,D_1e^{2r_1}/\Sigma_{11}\},\qquad
v_2:=e^{-2\beta}=\min\{1,D_2e^{2r_2}/\Sigma_{22}\}.
\end{equation*}
Then $\nu_\alpha=\sqrt{1-v_1}$ and $\nu_\beta=\sqrt{1-v_2}$. Given $(\alpha,\beta)\in\bbR_+^2$, consider the normalized covariance
\begin{equation}\label{eq:explicit-normalized-K}
\widetilde{\bSigma}_{\rmN}^{*}:=
\begin{cases}
\begin{bmatrix}v_1&\rho_{\bSigma}-\nu_\alpha\nu_\beta\\
\rho_{\bSigma}-\nu_\alpha\nu_\beta&v_2\end{bmatrix},
&\text{if }(\alpha,\beta)\in\calP_1,\\[1.2em]
\begin{bmatrix}v_1&0\\0&v_2\end{bmatrix},
&\text{if }(\alpha,\beta)\in\calP_2,\\[1.2em]
\begin{bmatrix}v_1&\rho_{\bSigma}v_1\\
\rho_{\bSigma}v_1&1-\rho_{\bSigma}^2+\rho_{\bSigma}^2v_1\end{bmatrix},
&\text{if }(\alpha,\beta)\in\calP_3,\\[1.2em]
\begin{bmatrix}1-\rho_{\bSigma}^2+\rho_{\bSigma}^2v_2&\rho_{\bSigma}v_2\\
\rho_{\bSigma}v_2&v_2\end{bmatrix},
&\text{if }(\alpha,\beta)\in\calP_4.
\end{cases}
\end{equation}
Here the regime definitions are evaluated with correlation $\rho_{\bSigma}$. Set $\bSigma_{\rmN}^{*}:=\bt_{\bSigma}\widetilde{\bSigma}_{\rmN}^{*}\bt_{\bSigma}$ and $\widetilde{\bSigma}:=\bt_{\bSigma}^{-1}\bSigma\bt_{\bSigma}^{-1}$. The latter matrix has unit diagonal and off-diagonal entry $\rho_{\bSigma}$.

On $\calP_1$, subtraction yields
\begin{equation*}
\widetilde{\bSigma}-\widetilde{\bSigma}_{\rmN}^{*}
=\begin{bmatrix}\nu_\alpha\\\nu_\beta\end{bmatrix}
\begin{bmatrix}\nu_\alpha&\nu_\beta\end{bmatrix}\succeq0.
\end{equation*}
The diagonal entries satisfy $[\widetilde{\bSigma}_{\rmN}^{*}]_{11}=v_1$ and $[\widetilde{\bSigma}_{\rmN}^{*}]_{22}=v_2$. Here the residual correlation is $r_{\alpha,\beta}=(\rho_{\bSigma}-\nu_\alpha\nu_\beta)/\sqrt{v_1v_2}$. The $\calP_1$ inequalities imply $0\leq r_{\alpha,\beta}<1$, as verified in Appendix~\ref{app:regularity}. Therefore,  $|\widetilde{\bSigma}_{\rmN}^{*}|=v_1v_2(1-r_{\alpha,\beta}^2)>0$.

On $\calP_2$, $\widetilde{\bSigma}_{\rmN}^{*}=\diag(v_1,v_2)\succ0$ satisfies the diagonal bounds. The difference $\widetilde{\bSigma}-\widetilde{\bSigma}_{\rmN}^{*}$ is positive semidefinite because its diagonal entries satisfy $(\nu_\alpha^2,\nu_\beta^2)\in\bbR_{+,0}^2$ and its determinant is $\nu_\alpha^2\nu_\beta^2-\rho_{\bSigma}^2\geq0$. On $\calP_3$, subtraction instead leads to
\begin{equation*}
\widetilde{\bSigma}-\widetilde{\bSigma}_{\rmN}^{*}
=(1-v_1)\begin{bmatrix}1\\\rho_{\bSigma}\end{bmatrix}
\begin{bmatrix}1&\rho_{\bSigma}\end{bmatrix}\succeq0.
\end{equation*}
The first diagonal entry of $\widetilde{\bSigma}_{\rmN}^{*}$ equals $v_1$, and the second equals $1-\rho_{\bSigma}^2\nu_\alpha^2\leq1-\nu_\beta^2=v_2$, because $\rho_{\bSigma}>\nu_\beta/\nu_\alpha$. Also, $|\widetilde{\bSigma}_{\rmN}^{*}|=v_1(1-\rho_{\bSigma}^2)>0$. The same calculations with the two coordinates interchanged prove feasibility and positive definiteness on $\calP_4$. 
Since $\bt_{\bSigma}$ is invertible, $\bx^{\rvT}\bSigma_{\rmN}^{*}\bx=(\bt_{\bSigma}\bx)^{\rvT}\widetilde{\bSigma}_{\rmN}^{*}(\bt_{\bSigma}\bx)>0$ for every nonzero $\bx\in\bbR^2$. Similarly, $\bSigma-\bSigma_{\rmN}^{*}=\bt_{\bSigma}(\widetilde{\bSigma}-\widetilde{\bSigma}_{\rmN}^{*})\bt_{\bSigma}\succeq0$, and
\begin{equation*}
[\bSigma_{\rmN}^{*}]_{jj}=\Sigma_{jj}[\widetilde{\bSigma}_{\rmN}^{*}]_{jj}\leq\Sigma_{jj}v_j=\min\{\Sigma_{jj},D_je^{2r_j}\}.
\end{equation*}

Directly expanding the determinants leads to
\begin{equation}\label{eq:explicit-K-determinants}
|\widetilde{\bSigma}_{\rmN}^{*}|=
\begin{cases}
v_1v_2(1-r_{\alpha,\beta}^2),&\text{if }(\alpha,\beta)\in\calP_1,\\
v_1v_2,&\text{if }(\alpha,\beta)\in\calP_2,\\
v_1(1-\rho_{\bSigma}^2),&\text{if }(\alpha,\beta)\in\calP_3,\\
v_2(1-\rho_{\bSigma}^2),&\text{if }(\alpha,\beta)\in\calP_4.
\end{cases}
\end{equation}
At the $\calP_1$--$\calP_2$ interface, $\rho_{\bSigma}=\nu_\alpha\nu_\beta$, so both off-diagonal entries are zero. At the $\calP_1$--$\calP_3$ interface, $\nu_\beta=\rho_{\bSigma}\nu_\alpha$, so the first candidate's off-diagonal entry is $\rho_{\bSigma}v_1$ and its second diagonal is $v_2=1-\rho_{\bSigma}^2+\rho_{\bSigma}^2v_1$. The $\calP_1$--$\calP_4$ interface follows by interchanging the coordinates. Hence, the matrices agree on all shared finite interfaces.

If $\alpha>0$ and $\beta=0$, use the $\calP_3$ matrix with $v_2=1$. Its determinant is $v_1(1-\rho_{\bSigma}^2)$ and its common-rate value is $\alpha$. If $\alpha=0$ and $\beta>0$, use the $\calP_4$ matrix with $v_1=1$; its common-rate value is $\beta$. In either case, the preceding calculations imply $\bzero_{2\times2}\prec\bSigma_{\rmN}^{*}\preceq\bSigma$ and $[\bSigma_{\rmN}^{*}]_{jj}\leq D_je^{2r_j}$ for $j\in[2]$. If $\alpha=\beta=0$, take $\bSigma_{\rmN}^{*}=\bSigma$, with common rate zero. These choices agree with the limits of the formulas for $(\alpha,\beta)\in\bbR_+^2$ and with~\eqref{eq:Upsilon-axes}.

Since $|\bSigma|=\Sigma_{11}\Sigma_{22}(1-\rho_{\bSigma}^2)$ and $|\bSigma_{\rmN}^{*}|=\Sigma_{11}\Sigma_{22}|\widetilde{\bSigma}_{\rmN}^{*}|$, determinant scaling and~\eqref{eq:explicit-K-determinants} imply
\begin{equation}\label{eq:pf-explicit-rate-optimal}
\frac12\log\frac{|\bSigma|}{|\bSigma_{\rmN}^{*}|}
=\frac12\log\frac{1-\rho_{\bSigma}^2}{|\widetilde{\bSigma}_{\rmN}^{*}|}
=\Upsilon_{\rho_{\bSigma}}(\alpha,\beta).
\end{equation}
Here $-\tfrac12\log v_1=\alpha$ and $-\tfrac12\log v_2=\beta$.

To prove optimality, take any positive-definite feasible $\bSigma_{\rmN}$ and independent Gaussian vectors $(\bW,\bN)$ with covariances $(\bSigma-\bSigma_{\rmN},\bSigma_{\rmN})$. Let $(X,Y)^{\rvT}=\bW+\bN$. Feasibility implies $I(X;\bW)=\tfrac12\log(\Sigma_{11}/[\bSigma_{\rmN}]_{11})\geq\alpha$ and, similarly, $I(Y;\bW)\geq\beta$. Yu's envelope~\cite[Theorem~3]{yu2023} yields
\begin{equation*}
\frac12\log\frac{|\bSigma|}{|\bSigma_{\rmN}|}
=I(X,Y;\bW)\geq\Upsilon_{\rho_{\bSigma}}(\alpha,\beta).
\end{equation*}
The candidates attain the bound by~\eqref{eq:pf-explicit-rate-optimal}. Singular residuals have determinant zero, and strict concavity of $\log|\bSigma_{\rmN}|$ on the convex positive-definite feasible set implies uniqueness.

We next prove the covering assertion for a fixed retained cell $q\in\calJ_n$. Use the optimizer $\bSigma_{\mathrm N,q}\in\calF_q$ and define $\bb_q:=(\bSigma_q-\bSigma_{\mathrm N,q})^{1/2}$ and $I_q:=\tfrac12\log(|\bSigma_q|/|\bSigma_{\mathrm N,q}|)$. The bound~\eqref{eq:optimizer-compact} makes all constants below independent of $q$, the nonnegative design rates, and the shell point.

Introduce the Gaussian distribution $Q$ under which $\bU_q\sim\calN(\bzero_2,\bone_2)$ and $\bN\sim\calN(\bzero_2,\bSigma_{\mathrm N,q})$ are independent and $\bX=\bb_q\bU_q+\bN$. Given $\bu,\bx\in\bbR^2$, taking the logarithm of the ratio of the conditional density of $\bX$ given $\bU_q=\bu$ to its marginal density leads to
\begin{equation*}
\imath(\bu;\bx):=I_q-\frac12(\bx-\bb_q\bu)^{\rvT}\bSigma_{\mathrm N,q}^{-1}(\bx-\bb_q\bu)
+\frac12\bx^{\rvT}\bSigma_q^{-1}\bx.
\end{equation*}
Given blocks $\bu^n=[\bu_1,\ldots,\bu_n]\in\bbR^{2\times n}$ and $\bx^n=[\bx_1,\ldots,\bx_n]\in\bbR^{2\times n}$, define $\imath(\bu^n;\bx^n):=\sum_{i\in[n]}\imath(\bu_i;\bx_i)$. Completing the square in the joint density, or applying Gaussian conditioning, leads to
\begin{equation*}
\bU_q\mid\bX=\bx\sim\calN(\bb_q^{\rvT}\bSigma_q^{-1}\bx,\bSigma_{\bU\mid\bX,q}),
\end{equation*}
where
\begin{equation*}
\bSigma_{\bU\mid\bX,q}=\bone_2-\bb_q^{\rvT}\bSigma_q^{-1}\bb_q
=(\bone_2+\bb_q^{\rvT}\bSigma_{\mathrm N,q}^{-1}\bb_q)^{-1}\succ0.
\end{equation*}
The last equality is the matrix inversion identity with $\bSigma_q=\bSigma_{\mathrm N,q}+\bb_q\bb_q^{\rvT}$. Its inverse expression proves positive definiteness even if $\bb_q$ is singular.

Given a fixed shell point $\bx^n\in\calT_n(\bSigma_q)$, use the product posterior $Q(\cdot\mid\bx^n):=\bigotimes_{i\in[n]}Q_{\bU_q\mid\bX=\bx_i}$, and let $\bU_q^n=[\bU_{q,1},\ldots,\bU_{q,n}]$ denote the resulting random block. Define $J:=\imath(\bU_q^n;\bx^n)$ and $E_j:=\Delta_j(\bU_q^n;\bx^n)$, where $\Delta_j(\bu^n;\bx^n):=\norm{(\bx^n-\bb_q\bu^n)_j}^2$. The following lemma concerns $Q(\cdot\mid\bx^n)$.

\begin{lemma}[Uniform posterior small-ball bound]\label{lem:small-ball}
Let $\bY:=(J,E_1,E_2)^{\rvT}$ under the product posterior at a fixed source block. There exist $c>0$ and an integer $n_0$ such that, for every $n\geq n_0$,
\begin{equation}\label{eq:small-ball-bound}
Q\left\{\norm{\bY-\bbE_Q[\bY\mid\bx^n]}\leq\frac{s_n}{64}\,\middle|\,\bx^n\right\}\geq\frac{c}{(\log n)^7},
\end{equation}
uniformly over $q\in\calJ_n$, all nonnegative design rates, and $\bx^n\in\calT_n(\bSigma_q)$.
\end{lemma}
Its proof follows the covering argument below.

Let $\bG$ be the Gram matrix of the posterior residual block $\bx^n-\bb_q\bU_q^n$. Gaussian conditioning implies that each column has residual mean $\bSigma_{\mathrm N,q}\bSigma_q^{-1}\bx$ and covariance $\bSigma_{\mathrm N,q}-\bSigma_{\mathrm N,q}\bSigma_q^{-1}\bSigma_{\mathrm N,q}$. Since $\bx^n(\bx^n)^{\rvT}=n\hat{\bSigma}(\bx^n)$,
\begin{align}\label{eq:posterior-gram-mean}
\bg_{n,q}:=\bbE_Q[\bG\mid\bx^n]
&=n\bSigma_{\mathrm N,q}+n\bSigma_{\mathrm N,q}\bSigma_q^{-1}
(\hat{\bSigma}(\bx^n)-\bSigma_q)\bSigma_q^{-1}\bSigma_{\mathrm N,q}.
\end{align}
Its $j$-th diagonal entry is $\bbE_Q[E_j\mid\bx^n]$. Summing the information densities leads to
\begin{equation}\label{eq:J-matrix}
J=nI_q-\frac12\tr(\bSigma_{\mathrm N,q}^{-1}\bG)
+\frac n2\tr(\bSigma_q^{-1}\hat{\bSigma}(\bx^n)),
\end{equation}
and hence,
\begin{equation*}
\bbE_Q[J\mid\bx^n]
=nI_q+\frac n2\tr[\bSigma_q^{-1}(\bSigma_q-\bSigma_{\mathrm N,q})
\bSigma_q^{-1}(\hat{\bSigma}(\bx^n)-\bSigma_q)].
\end{equation*}
The uniform operator bounds~\eqref{eq:optimizer-compact}, the trace Cauchy--Schwarz inequality~\cite[Corollary~IV.2.6]{bhatia1997}, and the shell condition $\norm{\hat{\bSigma}(\bx^n)-\bSigma_q}\leq C_0/n$ imply
\begin{equation*}
|\bbE_Q[J\mid\bx^n]-nI_q|
+\sum_{j\in[2]}|\bbE_Q[E_j\mid\bx^n]-n[\bSigma_{\mathrm N,q}]_{jj}|
\leq C,
\end{equation*}
with $C$ independent of the cell, rates, and shell point.

Define the simultaneous posterior event
\begin{equation*}
\calG_n(\bx^n):=\{\bu^n\in\bbR^{2\times n}:\imath(\bu^n;\bx^n)\leq nI_q+s_n/16,\quad
\Delta_j(\bu^n;\bx^n)\leq n[\bSigma_{\mathrm N,q}]_{jj}+s_n/16\text{ for }j\in[2]\}.
\end{equation*}
Choose $n$ large enough that $C\leq3s_n/64$. Conditioned on $\{\norm{\bY-\bbE_Q[\bY\mid\bx^n]}\leq s_n/64\}$,
\begin{equation*}
J\leq nI_q+C+s_n/64\leq nI_q+s_n/16.
\end{equation*}
Given $j\in[2]$,
\begin{equation*}
E_j\leq n[\bSigma_{\mathrm N,q}]_{jj}+C+s_n/64
\leq n[\bSigma_{\mathrm N,q}]_{jj}+s_n/16.
\end{equation*}
Since $J=\imath(\bU_q^n;\bx^n)$ and $E_j=\Delta_j(\bU_q^n;\bx^n)$,
\begin{equation*}
\left\{\norm{\bY-\bbE_Q[\bY\mid\bx^n]}\leq\frac{s_n}{64}\right\}
\subseteq\{\bU_q^n\in\calG_n(\bx^n)\}.
\end{equation*}
Therefore, Lemma~\ref{lem:small-ball} implies $Q\{\bU_q^n\in\calG_n(\bx^n)\mid\bx^n\}\geq c(\log n)^{-7}$.

Bayes' identity implies $dQ_{\bU_q}^{\otimes n}/dQ(\cdot\mid\bx^n)=e^{-J}$. The density ratio is defined for every fixed $\bx^n$, because both Gaussian densities are positive. Hence,
\begin{align*}
Q_{\bU_q}^{\otimes n}\{\calG_n(\bx^n)\}
&=\bbE_Q[e^{-J}\bbI_{\{\bU_q^n\in\calG_n(\bx^n)\}}\mid\bx^n]\\
&\geq e^{-nI_q-s_n/16}Q\{\calG_n(\bx^n)\mid\bx^n\}\\
&\geq c(\log n)^{-7}e^{-nI_q-s_n/16}=:p_{n,q}.
\end{align*}
Generate $M_{0,q}=\lceil e^{nI_q+s_n/8}\rceil$ independent standard Gaussian codewords. Given a fixed shell point, independence and $1-t\leq e^{-t}$ imply the missed-covering bound
\begin{equation*}
(1-p_{n,q})^{M_{0,q}}\leq e^{-p_{n,q}M_{0,q}}\leq\exp\{-c(\log n)^{-7}e^{s_n/16}\}.
\end{equation*}

To cover the shell, choose a maximal $\calN_n\subset\calT_n(\bSigma_q)$ with pairwise distances at least $\delta_n:=1/\log n$. The shell is closed and bounded because
\begin{equation*}
\norm{\bx^n}^2=n\tr\hat{\bSigma}(\bx^n)\leq n\tr\bSigma_q+\sqrt2 C_0\leq C^2n,
\end{equation*}
where $|\tr(\hat{\bSigma}(\bx^n)-\bSigma_q)|\leq\sqrt2\norm{\hat{\bSigma}(\bx^n)-\bSigma_q}$. Thus, $\calN_n$ is finite, and maximality places every shell point within distance $\delta_n$ of $\calN_n$. Its disjoint radius-$\delta_n/2$ open balls lie inside the radius-$C\sqrt n+\delta_n/2$ ball in $\bbR^{2n}$, so
\begin{equation*}
|\calN_n|\leq\left(\frac{C\sqrt n+\delta_n/2}{\delta_n/2}\right)^{2n}
=(1+2C\sqrt n\log n)^{2n}\leq e^{C'n\log n}.
\end{equation*}
By the union bound, the probability of missing a point of $\calN_n$ is at most
\begin{equation*}
\exp\{C'n\log n-c(\log n)^{-7}e^{s_n/16}\}\longrightarrow0.
\end{equation*}
Indeed, $s_n/16-\log n-8\log\log n\to\infty$. Hence, a codebook realization covers every point of $\calN_n$ for sufficiently large $n$.

Fix such a realization. Given $\bx^n\in\calT_n(\bSigma_q)$, choose $\bar\bx^n\in\calN_n$ with $\norm{\bx^n-\bar\bx^n}\leq\delta_n$ and a common codeword $\bu^n\in\calG_n(\bar\bx^n)$. The row-wise triangle inequality implies
\begin{align*}
\norm{(\bx^n-\bb_q\bu^n)_j}^2
&\leq\left(\sqrt{n[\bSigma_{\mathrm N,q}]_{jj}+s_n/16}+\delta_n\right)^2\\
&=n[\bSigma_{\mathrm N,q}]_{jj}+s_n/16+2\delta_n\sqrt{n[\bSigma_{\mathrm N,q}]_{jj}+s_n/16}+\delta_n^2\\
&\leq n[\bSigma_{\mathrm N,q}]_{jj}+(17/16+2\sqrt{\kappa^{-1}+1})s_n\\
&\leq n[\bSigma_{\mathrm N,q}]_{jj}+\chi_\kappa s_n=\xi_{j,q,n}^2,
\end{align*}
where the penultimate inequality uses $[\bSigma_{\mathrm N,q}]_{jj}\leq\kappa^{-1}$, $s_n/n\leq1$, $\delta_n\sqrt n=s_n$, and $\delta_n^2\leq s_n$ and the last uses~\eqref{eq:private-radius-constant}. Thus, every shell point has a common codeword with both residuals in the prescribed private balls.

Apply Lemma~\ref{lem:euclidean-uniform} to $\calB_n(\bzero_n,\xi_{j,q,n})$ for $j\in[2]$. The lemma implies the existence of deterministic private offset codebooks satisfying
\begin{align*}
\log M_{j,q}
&\leq\log C_{\mathrm{cov}}+4\log n+\frac n2\left[\log\frac{[\bSigma_{\mathrm N,q}]_{jj}+\chi_\kappa s_n/n}{D_j}\right]^+\\
&\leq\frac n2\left[\log\frac{[\bSigma_{\mathrm N,q}]_{jj}}{D_j}\right]^+
+\frac{\chi_\kappa}{2\kappa}s_n+4\log n+\log C_{\mathrm{cov}}.
\end{align*}
The second inequality uses $[a+b]^+\leq[a]^++b$ for $b\geq0$, $\log(1+t)\leq t$, and $[\bSigma_{\mathrm N,q}]_{jj}\geq\kappa$. The maps $f_q$ and $\phi_{j,q}$ in Section~\ref{sec:code-construction} select a common codeword and private offsets with squared errors at most $nD_j$, proving~\eqref{eq:shell-distortion-guarantee}. Finally, $I_q\geq0$ and $\lceil e^t\rceil\leq2e^t$ for $t\geq0$ imply $\log M_{0,q}\leq nI_q+s_n/8+\log2$. Enlarging the common constant proves~\eqref{eq:shell-M0}--\eqref{eq:shell-Mj}. Since $\bb_q$ was never inverted, the proof includes rank-deficient $\bSigma_q-\bSigma_{\mathrm N,q}$.

\subsection{Proof of Lemma~\ref{lem:small-ball}}
\label{app:small-ball-proof}
Recall the Gaussian posterior and its positive-definite covariance $\bSigma_{\bU\mid\bX,q}$ from the proof above. Let $\bZ^n\in\bbR^{2\times n}$ have independent standard Gaussian entries. The product posterior can be represented as
\begin{equation*}
\bU_q^n=\bb_q^{\rvT}\bSigma_q^{-1}\bx^n+\bSigma_{\bU\mid\bX,q}^{1/2}\bZ^n.
\end{equation*}
The fixed shell point satisfies $\hat{\bSigma}(\bx^n)\succeq(\kappa-C_0/n)\bone_2\succeq(\kappa/2)\bone_2$ for sufficiently large $n$. Thus, $n^{-1/2}\hat{\bSigma}(\bx^n)^{-1/2}\bx^n$ has two orthonormal rows. Extend those rows to an orthogonal basis of $\bbR^n$ to choose an orthogonal matrix $\bo\in\bbR^{n\times n}$ such that
\begin{equation*}
\bx^n \bo=\sqrt n\,\hat{\bSigma}(\bx^n)^{1/2}[\bone_2\ \bzero_{2\times(n-2)}].
\end{equation*}
Right multiplication by $\bo$ preserves every residual Gram matrix, and $\bZ^n\bo$ has the same distribution as $\bZ^n$ by orthogonal invariance of independent standard Gaussian rows. Therefore, we may use the displayed form for $\bx^n$ without changing the conditional distribution of $\bG$, $J$, or $(E_1,E_2)$.

Let $\bZ^n=[\bA\ \bB]$, where $\bA\in\bbR^{2\times2}$ and $\bB\in\bbR^{2\times(n-2)}$ are independent standard Gaussian matrices. Define the deterministic matrices $\bh_q:=\bSigma_{\mathrm N,q}\bSigma_q^{-1}$ and $\bd_q:=\bb_q\bSigma_{\bU\mid\bX,q}^{1/2}$, both in $\bbR^{2\times2}$. The posterior residual is
\begin{equation*}
\bN^n:=\bx^n-\bb_q\bU_q^n=\bh_q\bx^n-\bd_q\bZ^n.
\end{equation*}
The expression for $\bN^n$ follows from $\bone_2-\bb_q\bb_q^{\rvT}\bSigma_q^{-1}=\bSigma_{\mathrm N,q}\bSigma_q^{-1}$. All coefficient norms are uniformly bounded: for example, $\norm{\bh_q}_{\mathrm{op}}\leq\kappa^{-2}$, $\norm{\bd_q}_{\mathrm{op}}\leq\kappa^{-1/2}$, and $\norm{\bSigma_{\mathrm N,q}^{-1}}\leq\sqrt2\kappa^{-1}$. The bound for $\bd_q$ uses $\bSigma_{\bU\mid\bX,q}\preceq\bone_2$ and $\bb_q\bb_q^{\rvT}\preceq\bSigma_q\preceq\kappa^{-1}\bone_2$.

Expanding $\bG=\bN^n(\bN^n)^{\rvT}$ and subtracting its expectation leads to
\begin{align}
\bG-\bg_{n,q}
&=-\bh_q\bx^n(\bZ^n)^{\rvT}\bd_q^{\rvT}-\bd_q\bZ^n(\bx^n)^{\rvT}\bh_q^{\rvT}\notag\\
&\quad+\bd_q[\bZ^n(\bZ^n)^{\rvT}-n\bone_2]\bd_q^{\rvT}.
\label{eq:gram-centered-expansion}
\end{align}
Here $\bbE[\bZ^n]=0$ and $\bbE[\bZ^n(\bZ^n)^{\rvT}]=n\bone_2$. The diagonal entries of $\bG-\bg_{n,q}$ are the centered residual errors, and~\eqref{eq:J-matrix} implies $J-\bbE_Q[J\mid\bx^n]=-\tfrac12\tr[\bSigma_{\mathrm N,q}^{-1}(\bG-\bg_{n,q})]$. The trace Cauchy--Schwarz inequality yields
\begin{equation*}
\norm{\bY-\bbE_Q[\bY\mid\bx^n]}
\leq\sqrt{1+\tfrac14\norm{\bSigma_{\mathrm N,q}^{-1}}^2}\,\norm{\bG-\bg_{n,q}}.
\end{equation*}
The comparison between $\norm{\bY-\bbE_Q[\bY\mid\bx^n]}$ and $\norm{\bG-\bg_{n,q}}$ remains valid when the covariance of $\bY$ is singular.

Given a fixed $a\in(0,1)$, to be chosen using the uniform coefficient bounds, define
\begin{equation*}
\calD_{n,a}:=\{\norm{\bA}\leq a/\log n,\quad
\norm{\bB\bB^{\rvT}-(n-2)\bone_2}\leq as_n\}.
\end{equation*}
On $\calD_{n,a}$, the rotated form of $\bx^n$ implies
\begin{equation*}
\norm{\bx^n(\bZ^n)^{\rvT}}
=\sqrt n\norm{\hat{\bSigma}(\bx^n)^{1/2}\bA^{\rvT}}
\leq C\sqrt n\norm{\bA}\leq Ca s_n.
\end{equation*}
Also, splitting $\bZ^n$ into its two blocks leads to
\begin{align*}
\norm{\bZ^n(\bZ^n)^{\rvT}-n\bone_2}
&\leq\norm{\bA\bA^{\rvT}-2\bone_2}+\norm{\bB\bB^{\rvT}-(n-2)\bone_2}\\
&\leq\norm{\bA}^2+2\sqrt2+as_n\leq C+as_n.
\end{align*}
Using $\norm{\bA\bA^{\rvT}}\leq\norm{\bA}^2$ and~\eqref{eq:gram-centered-expansion}, we obtain $\norm{\bY-\bbE_Q[\bY\mid\bx^n]}\leq C_1as_n+C_2$ on $\calD_{n,a}$, with uniform constants. Choose $C_1a\leq1/128$ and then $n$ large enough that $C_2\leq s_n/128$. Consequently,
\begin{equation*}
\calD_{n,a}\subseteq\left\{\norm{\bY-\bbE_Q[\bY\mid\bx^n]}\leq\frac{s_n}{64}\right\}.
\end{equation*}

The four-dimensional standard Gaussian density of $\bA$ is at least $(2\pi)^{-2}e^{-1/2}$ on the unit ball. Integrating over the radius-$a/\log n$ ball leads to
\begin{equation*}
\Pr\{\norm{\bA}\leq a/\log n\}\geq c_a(\log n)^{-4}.
\end{equation*}

Let $m:=n-2$ and $B_{ji}:=[\bB]_{ji}$ for $j\in[2]$ and $i\in[m]$. Define the independent identically distributed vectors
\begin{equation*}
\bV_i:=(B_{1i}^2-1,B_{1i}B_{2i},B_{2i}^2-1)^{\rvT},\qquad i\in[m].
\end{equation*}
Gaussian independence and $\bbE[B_{ji}^4]=3$ imply that $\bV_i$ has mean zero and covariance $\bSigma_{\bV}:=\diag(2,1,2)$. The third absolute moment is finite because the entries are quadratic Gaussian polynomials. The convex-set Berry--Esseen bound~\cite[Theorem~1.1]{bentkus2005}, applied to $\bSigma_{\bV}^{-1/2}\bV_i$, implies that a constant $C$ exists such that
\begin{equation*}
\left|\Pr\left\{m^{-1/2}\sum_{i\in[m]}\bV_i\in\calA\right\}
-\Pr\{\bZ_{\mathrm G}\in\calA\}\right|\leq C/\sqrt m
\end{equation*}
uniformly over convex Borel sets $\calA\subset\bbR^3$, where $\bZ_{\mathrm G}\sim\calN(\bzero_3,\bSigma_{\bV})$. 
The two diagonal entries and the upper off-diagonal entry of $\bB\bB^{\rvT}-m\bone_2$ are the coordinates of $\sum_{i\in[m]}\bV_i$. The squared Frobenius norm counts the off-diagonal entry twice. Define
\begin{equation*}
\calA_{n,a}:=\{\bv\in\bbR^3:[\bv]_1^2+2[\bv]_2^2+[\bv]_3^2\leq a^2s_n^2/m\}.
\end{equation*}
Then
\begin{equation*}
\Big\{\norm{\bB\bB^{\rvT}-m\bone_2}\leq as_n\Big\}
=\left\{m^{-1/2}\sum_{i\in[m]}\bV_i\in\calA_{n,a}\right\}.
\end{equation*}

Let $t_n:=as_n/\sqrt m$. For sufficiently large $n$, $t_n\leq1$, so $\calA_{n,a}$ lies in the unit ball, on which the Gaussian density of $\bZ_{\mathrm G}$ has a positive minimum. The volume of $\calA_{n,a}$ is $(4\pi/(3\sqrt2))t_n^3$, by the linear change of variables $\bv\mapsto([\bv]_1,\sqrt2[\bv]_2,[\bv]_3)$. Since $s_n/\sqrt m=\sqrt{n/(n-2)}/\log n\geq1/\log n$, the density and volume bounds imply
\begin{equation*}
\Pr\{\bZ_{\mathrm G}\in\calA_{n,a}\}\geq c_a'(\log n)^{-3}.
\end{equation*}
Consequently,
\begin{align*}
\Pr\Big\{\norm{\bB\bB^{\rvT}-m\bone_2}\leq as_n\Big\}
&\geq c_a'(\log n)^{-3}-C/\sqrt m\\
&\geq\tfrac12c_a'(\log n)^{-3}
\end{align*}
for sufficiently large $n$, since $(\log n)^3/\sqrt{n-2}\to0$. Independence of $\bA$ and $\bB$, together with the inclusion of $\calD_{n,a}$ above, implies
\begin{equation*}
Q\left\{\norm{\bY-\bbE_Q[\bY\mid\bx^n]}\leq\frac{s_n}{64}\,\middle|\,\bx^n\right\}
\geq Q(\calD_{n,a}\mid\bx^n)\geq\frac{c}{(\log n)^7},
\end{equation*}
proving~\eqref{eq:small-ball-bound}.
\section{Proof of~\eqref{eq:En-prob}}
\label{app:taylor-remainder}
Fix $\zeta>0$. Given $\bSigma\in\calK$ and $(r_1,r_2)\in\bbR_{+,0}^2$, let $\Delta(\bSigma,r_1,r_2):=\norm{\bSigma-\bSigma_{XY}}+\sum_{j\in[2]}|r_j-R_j^*|$ and
\begin{equation*}
\psi(\bSigma,r_1,r_2):=\rvR_0(r_1,r_2,D_1,D_2\mid Q_{XY}^{\bSigma})-R_0^*
-\ip{\bg}{\bSigma-\bSigma_{XY}}
+\sum_{j\in[2]}\lambda_j^*(r_j-R_j^*).
\end{equation*}
Given any $a>0$, Proposition~\ref{prop:regularity} and~\eqref{eq:local-differential} imply the existence of $\delta_a>0$ such that
\begin{align}
\label{eq:app-taylor-small}
|\psi(\bSigma,r_1,r_2)|
&\leq a\Delta(\bSigma,r_1,r_2)
\end{align}
whenever $\Delta(\bSigma,r_1,r_2)\leq\delta_a$. At $R_j^*=0$, use the one-sided expansion along the nonnegative rates in~\eqref{eq:rjn-design} and Theorem~\ref{thm:main}.

Fix $B>0$, put $C_B:=B+1+|L_1|+|L_2|$, and let $\calG_{n,B}:=\{\sqrt n\norm{\hat{\bSigma}_n-\bSigma_{XY}}\leq B\}$. For sufficiently large $n$, $B/\sqrt n\leq\eta_{\mathrm{cov}}$ implies $\calG_{n,B}\subseteq\{\hat{\bSigma}_n\in\calK_0\}$. By~\eqref{eq:quantizer-error},
\begin{align}
\label{eq:app-quantization-difference}
\norm{\mathfrak{q}_n(\hat{\bSigma}_n)-\hat{\bSigma}_n}
&\leq\frac1n,
\end{align}
so $\mathfrak{q}_n(\hat{\bSigma}_n)\in\calK$ on $\calG_{n,B}$ for sufficiently large $n$. The triangle inequality and $\sqrt n(r_{j,n}-R_j^*)=L_j$ yield, on $\calG_{n,B}$,
\begin{align}
\label{eq:app-displacement-bound}
\Delta(\mathfrak{q}_n(\hat{\bSigma}_n),r_{1,n},r_{2,n})
&\leq\frac{B}{\sqrt n}+\frac1n
+\sum_{j\in[2]}\frac{|L_j|}{\sqrt n}\nonumber\\
&\leq\frac{C_B}{\sqrt n}.
\end{align}
Since $L_j-\sqrt n(r_{j,n}-R_j^*)=0$, substituting the definition of $\psi$ into~\eqref{eq:global-Taylor-remainder} leads to the following identity on $\calG_{n,B}$:
\begin{align}
E_n
&=\sqrt n\psi(\mathfrak{q}_n(\hat{\bSigma}_n),r_{1,n},r_{2,n})
+\sqrt n\ip{\bg}{\mathfrak{q}_n(\hat{\bSigma}_n)-\hat{\bSigma}_n}.
\label{eq:app-exact-cancellation}
\end{align}
The Cauchy--Schwarz inequality and~\eqref{eq:app-quantization-difference}
yield
\begin{align}
\label{eq:app-quantization-control}
\sqrt n\left|
\ip{\bg}{\mathfrak{q}_n(\hat{\bSigma}_n)-\hat{\bSigma}_n}
\right|
&\leq\frac{\norm{\bg}}{\sqrt n}.
\end{align}
Set $a:=\zeta/(4C_B)$. For all sufficiently large $n$,
$C_B/\sqrt n\leq\delta_a$ and $\norm{\bg}/\sqrt n\leq\zeta/4$.
Combining~\eqref{eq:app-taylor-small},
\eqref{eq:app-displacement-bound},
\eqref{eq:app-exact-cancellation}, and
\eqref{eq:app-quantization-control} leads to the following bound on $\calG_{n,B}$:
\begin{align}
\label{eq:app-remainder-small}
|E_n|
&\leq\sqrt n\,a\Delta(
\mathfrak{q}_n(\hat{\bSigma}_n),r_{1,n},r_{2,n})
+\frac{\norm{\bg}}{\sqrt n}\nonumber\\
&\leq aC_B+\frac{\norm{\bg}}{\sqrt n}
\leq\frac{\zeta}{2}<\zeta.
\end{align}
Thus, given fixed $(B,\zeta)\in\bbR_+^2$, the following inclusion holds for sufficiently large $n$:
\begin{align}
\label{eq:app-tail-inclusion}
\{|E_n|>\zeta\}&\subseteq\calG_{n,B}^{\mathrm c}.
\end{align}
The definition after~\eqref{eq:global-Taylor-remainder} implies $\{|E_n|>\zeta\}\cap\{\hat{\bSigma}_n\notin\calK_0\}=\varnothing$.

Let $\bH_i:=\bX_i\bX_i^{\rvT}-\bSigma_{XY}$. Independence and the Gaussian identity $\bbE_{P_{XY}}[[\bX]_a^2[\bX]_b^2]=[\bSigma_{XY}]_{aa}[\bSigma_{XY}]_{bb}+2[\bSigma_{XY}]_{ab}^2$, $a,b\in[2]$, imply
\begin{align}
\bbE_{P_{XY}^n}\big[n\norm{\hat{\bSigma}_n-\bSigma_{XY}}^2\big]
&=\frac1n\sum_{i\in[n]}\bbE_{P_{XY}}[\norm{\bH_i}^2]\notag\\
&=\bbE_{P_{XY}}[(\bX^{\rvT}\bX)^2]-\tr(\bSigma_{XY}^2)\notag\\
&=(\tr\bSigma_{XY})^2+\tr(\bSigma_{XY}^2)=:C_{\mathrm{mom}}.
\label{eq:app-second-moment}
\end{align}
Markov's inequality implies
\begin{align}
P_{XY}^n(\calG_{n,B}^{\mathrm c})
&=P_{XY}^n\bigg\{n\norm{\hat{\bSigma}_n-\bSigma_{XY}}^2>B^2\bigg\}
\leq\frac{C_{\mathrm{mom}}}{B^2}.
\label{eq:app-covariance-tail}
\end{align}
Given any fixed $B>0$, \eqref{eq:app-tail-inclusion} and~\eqref{eq:app-covariance-tail} yield
\begin{align}
\label{eq:remainder-tail-control}
0\leq\limsup_{n\to\infty}P_{XY}^n\{|E_n|>\zeta\}
&\leq\frac{C_{\mathrm{mom}}}{B^2}.
\end{align}
Letting $B\to\infty$ proves~\eqref{eq:En-prob}.

\section{Proof of Lemma~\ref{lem:entropy-support}}
\label{app:entropy-identification}
Fix $\blambda^*=(\lambda_1^*,\lambda_2^*)\in(0,1]^2$. With $I_Q$ evaluated under $P_{XY}Q_{W\mid XY}$, Yu's formula~\cite[Theorem~3]{yu2023} implies
\begin{equation*}
\Upsilon_{\rho_{XY}}(\alpha,\beta)
=\inf_{\substack{Q_{W|XY}:\ I_Q(X;W)\geq\alpha,\ I_Q(Y;W)\geq\beta}}
I_Q(X,Y;W),\qquad (\alpha,\beta)\in\bbR_{+,0}^2.
\end{equation*}
Fix $(\alpha,\beta)\in\bbR_{+,0}^2$ and $0<t<1$. Choose conditional distributions in the infimum for $(\alpha^*,\beta^*)$ and $(\alpha,\beta)$, select between them with probabilities $1-t$ and $t$ independently of $(X,Y)$, and include the choice in $W$. The three mutual informations are then the corresponding weighted averages. Letting both objective values approach their infima leads to
\begin{equation*}
\Upsilon_{\rho_{XY}}\big((1-t)\alpha^*+t\alpha,(1-t)\beta^*+t\beta\big)
\leq(1-t)R_0^*+t\Upsilon_{\rho_{XY}}(\alpha,\beta).
\end{equation*}
Subtract $R_0^*$, divide by $t>0$, and let $t\downarrow0$. The derivatives in Proposition~\ref{prop:regularity} yield
\begin{equation*}
\Upsilon_{\rho_{XY}}(\alpha,\beta)\geq R_0^*+\lambda_1^*(\alpha-\alpha^*)+\lambda_2^*(\beta-\beta^*),\qquad (\alpha,\beta)\in\bbR_{+,0}^2,
\end{equation*}

Given $\bSigma_{\rmN}$ satisfying $\bzero_{2\times2}\prec\bSigma_{\rmN}\preceq\bSigma_{XY}$, let $\alpha_{\bSigma_{\rmN}}:=\frac12\log\frac{\sigma_X^2}{[\bSigma_{\rmN}]_{11}}$, $\beta_{\bSigma_{\rmN}}:=\frac12\log\frac{\sigma_Y^2}{[\bSigma_{\rmN}]_{22}}$, and $t_{\bSigma_{\rmN}}:=\frac12\log\frac{|\bSigma_{XY}|}{|\bSigma_{\rmN}|}$.
The independent Gaussian decomposition $(X,Y)^{\rvT}=\bW+\bN$ with residual covariance $\bSigma_{\rmN}$ realizes these three mutual informations, even if the common covariance is singular. Hence,
\begin{align}
t_{\bSigma_{\rmN}}
&\geq\Upsilon_{\rho_{XY}}(\alpha_{\bSigma_{\rmN}},\beta_{\bSigma_{\rmN}})\notag\\
&\geq R_0^*+\lambda_1^*(\alpha_{\bSigma_{\rmN}}-\alpha^*)
+\lambda_2^*(\beta_{\bSigma_{\rmN}}-\beta^*).\label{eq:entropy-plane}
\end{align}
Using~\eqref{eq:alpha-star}, rearrange~\eqref{eq:entropy-plane} as
\begin{align}
\frac12\log\frac{|\bSigma_{XY}|}{|\bSigma_{\rmN}|}
+\sum_{j\in[2]}\frac{\lambda_j^*}{2}\log\frac{[\bSigma_{\rmN}]_{jj}}{D_j}
&\geq R_0^*+\lambda_1^*R_1^*+\lambda_2^*R_2^*.\label{eq:pf-entropy-objective-support}
\end{align}
At the optimizer $\bSigma_{\rmN}^{*}$ for $(\bSigma_{XY},R_1^*,R_2^*)$, the inequalities $[\bSigma_{\rmN}^{*}]_{jj}\leq D_je^{2R_j^*}$ imply $\alpha_{\bSigma_{\rmN}^{*}}\geq\alpha^*$ and $\beta_{\bSigma_{\rmN}^{*}}\geq\beta^*$, while~\eqref{eq:pf-explicit-rate-optimal} implies $t_{\bSigma_{\rmN}^{*}}=R_0^*$. Thus, \eqref{eq:entropy-plane} yields
\begin{equation*}
0\geq\lambda_1^*(\alpha_{\bSigma_{\rmN}^{*}}-\alpha^*)+\lambda_2^*(\beta_{\bSigma_{\rmN}^{*}}-\beta^*)\geq0.
\end{equation*}
Since $(\lambda_1^*,\lambda_2^*)\in\bbR_+^2$, both differences vanish. Hence, equality holds in~\eqref{eq:pf-entropy-objective-support}, proving~\eqref{eq:entropy-support-value}, also when a private rate is zero.

Given $\bSigma\in\calK$ sufficiently near $\bSigma_{XY}$, evaluate~\eqref{eq:entropy-support} at the optimizer $\bSigma_{\mathrm N,\bSigma}^{*}$ for private rates $(R_1^*,R_2^*)$. Since $[\bSigma_{\mathrm N,\bSigma}^{*}]_{jj}\leq D_je^{2R_j^*}$,
\begin{align}
\underline F_{\blambda^*}(\bSigma)
&\leq\frac12\log\frac{|\bSigma|}{|\bSigma_{\mathrm N,\bSigma}^{*}|}
+\sum_{j\in[2]}\frac{\lambda_j^*}{2}\log\frac{[\bSigma_{\mathrm N,\bSigma}^{*}]_{jj}}{D_j}\notag\\
&\leq\rvR_0(R_1^*,R_2^*,D_1,D_2\mid Q_{XY}^{\bSigma})
+\lambda_1^*R_1^*+\lambda_2^*R_2^*.\label{eq:entropy-upper-tangent}
\end{align}
Equality holds at $\bSigma_{XY}$ by~\eqref{eq:entropy-support-value}. To verify finiteness on $\bbS_{++}^2$, Hadamard's bound, $[\bSigma_{\rmN}]_{jj}\leq\Sigma_{jj}$, and $0<\lambda_j^*\leq1$ imply
\begin{equation*}
\underline F_{\blambda^*}(\bSigma)
\geq\frac12\log|\bSigma|
-\sum_{j\in[2]}\frac{1-\lambda_j^*}{2}\log\Sigma_{jj}
-\sum_{j\in[2]}\frac{\lambda_j^*}{2}\log D_j> -\infty,
\end{equation*}
while taking $\bSigma_{\rmN}=\bSigma$ leads to a finite upper bound.

The finite function $\varphi_{\blambda^*,\bSigma_{XY}}$ is concave by~\cite[Proposition~5]{liu2020}, so a matrix $\bg_{\mathrm{ent},\blambda^*}$ satisfies~\eqref{eq:supergradient-definition} at $\bSigma_{XY}$. Negating~\eqref{eq:entropy-upper-tangent} and using~\eqref{eq:converse-covariance-functional}, the derivatives of $\frac12\log|\bSigma|$ and $-\frac12\tr(\bSigma_{XY}^{-1}\bSigma)$ cancel at $\bSigma_{XY}$. Therefore, given a sufficiently small symmetric $\bh$, equality at $\bSigma_{XY}$ and~\eqref{eq:supergradient-definition} imply
\begin{equation*}
-\ip{\bg}{\bh}+o(\norm{\bh})
\leq\varphi_{\blambda^*,\bSigma_{XY}}(\bSigma_{XY}+\bh)-\varphi_{\blambda^*,\bSigma_{XY}}(\bSigma_{XY})
\leq\ip{\bg_{\mathrm{ent},\blambda^*}}{\bh}.
\end{equation*}
Apply the bounds to $t\bh$, divide by $t>0$, and let $t\downarrow0$ to obtain $-\ip{\bg}{\bh}\leq\ip{\bg_{\mathrm{ent},\blambda^*}}{\bh}$. Replacing $\bh$ by $-\bh$ implies equality, so every matrix satisfying~\eqref{eq:supergradient-definition} equals $-\bg$. Substitution into the displayed bounds yields
\begin{equation*}
\varphi_{\blambda^*,\bSigma_{XY}}(\bSigma_{XY}+\bh)-\varphi_{\blambda^*,\bSigma_{XY}}(\bSigma_{XY})
=-\ip{\bg}{\bh}+o(\norm{\bh}).
\end{equation*}
Thus, $\varphi_{\blambda^*,\bSigma_{XY}}$ is differentiable at $\bSigma_{XY}$ with gradient $-\bg$, and~\eqref{eq:converse-covariance-functional} implies~\eqref{eq:entropy-support-gradient}.

Finally, substituting $\bg_{\mathrm{ent},\blambda^*}=-\bg$ in~\eqref{eq:converse-fluctuation-definition} leads to
\begin{equation*}
g_{\blambda^*}(\bx)=-\bx^{\rvT}\bg\bx+\tr(\bg\bSigma_{XY})
=-[\jmath(x,y)-R_0^*].
\end{equation*}
Lemma~\ref{lem:critical-quantities} and~\eqref{eq:V-def} imply that the mean is zero and the variance is $V$, completing the proof.

\section{Proof of Corollary~\ref{cor:branchwise-region}}
\label{app:branchwise-region}
On $\calP_1$, let $r=r_{\alpha^*,\beta^*}$ at $\rho=\rho_{XY}$ and recall $u$ from~\eqref{eq:u-ratio}. Equations~\eqref{eq:pf-curved-slope}--\eqref{eq:pf-curved-slope-beta} and~\eqref{eq:Psi-gradient} imply
\begin{equation}\label{eq:c-D1}
\lambda_1^*=\frac{1-ru}{1-r^2}.
\end{equation}
\begin{equation}\label{eq:c2-D1}
\lambda_2^*=\frac{1-r/u}{1-r^2}.
\end{equation}
Substituting~\eqref{eq:c-D1}--\eqref{eq:c2-D1} and~\eqref{eq:V-branches} into~\eqref{eq:smooth-region} proves~\eqref{eq:branchwise-curved}. In $\calP_2,\calP_3,\calP_4$, equations~\eqref{eq:branch-derivatives} and~\eqref{eq:V-branches} imply that $(\lambda_1^*,\lambda_2^*,V)$ equals $(1,1,1)$, $(1,0,1/2)$, and $(0,1,1/2)$, respectively. Substitution into~\eqref{eq:smooth-region} leads to~\eqref{eq:branchwise-sum}--\eqref{eq:branchwise-one-source-Y}, retaining the zero-rate constraints.

At $r=0,u,u^{-1}$, the curved weights and dispersion coincide with those for $\calP_2,\calP_3,\calP_4$, respectively. Lemma~\ref{lem:boundary-classification} determines the admissible strict one-source boundary tuples. This completes the proof.

\bibliographystyle{IEEEtran}
\bibliography{references}

@article{bentkus2005,
  author = {Vidmantas Bentkus},
  title = {A {Lyapunov}-type bound in {$\mathbb R^d$}},
  journal = {Theory Probab. Appl.},
  volume = {49},
  number = {2},
  pages = {311--323},
  year = {2005}
}

@misc{cao2026image,
  author       = {Daming Cao},
  title        = {Second-Order Image Characterization with
                  Distortion and Secrecy},
  howpublished = {arXiv:2610.09472v1},
  month        = oct,
  year         = {2026}
}

@article{graywyner1974,
  author = {Robert M. Gray and Aaron D. Wyner},
  title = {Source coding for a simple network},
  journal = {Bell Syst. Tech. J.},
  volume = {53},
  number = {9},
  pages = {1681--1721},
  month = nov,
  year = {1974}
}

@article{heegardberger1985,
  author = {Chris Heegard and Toby Berger},
  title = {Rate distortion when side information may be absent},
  journal = {IEEE Trans. Inf. Theory},
  volume = {31},
  number = {6},
  pages = {727--734},
  month = nov,
  year = {1985}
}

@inproceedings{ingberkochman2011,
  author = {Amir Ingber and Yuval Kochman},
  title = {The dispersion of lossy source coding},
  booktitle = {Proc. Data Compression Conf.},
  pages = {53--62},
  month = mar,
  year = {2011},
  address = {Snowbird, UT, USA}
}

@article{kostinaverdu2012,
  author = {Victoria Kostina and Sergio Verd{\'u}},
  title = {Fixed-length lossy compression in the finite blocklength regime},
  journal = {IEEE Trans. Inf. Theory},
  volume = {58},
  number = {6},
  pages = {3309--3338},
  month = jun,
  year = {2012}
}

@inproceedings{lili2025,
  author = {Xiang Li and Cheuk Ting Li},
  title = {New second-order achievability bounds for coding with side information via type deviation convergence},
  booktitle = {Proc. IEEE Int. Symp. Inf. Theory},
  pages = {1--6},
  month = jun,
  year = {2026},
  address = {Guangzhou, China}
}

@article{liu2020,
  author = {Jingbo Liu and Thomas A. Courtade and Paul Cuff and Sergio Verd{\'u}},
  title = {Smoothing {Brascamp--Lieb} inequalities and strong converses of coding theorems},
  journal = {IEEE Trans. Inf. Theory},
  volume = {66},
  number = {2},
  pages = {704--721},
  month = feb,
  year = {2020}
}

@article{noingberweissman2016,
  author = {Albert No and Amir Ingber and Tsachy Weissman},
  title = {Strong successive refinability and rate-distortion-complexity tradeoff},
  journal = {IEEE Trans. Inf. Theory},
  volume = {62},
  number = {6},
  pages = {3618--3635},
  month = jun,
  year = {2016}
}

@article{nomurahan2014,
  author = {Ryo Nomura and Te Sun Han},
  title = {Second-order {Slepian--Wolf} coding theorems for non-mixed and mixed sources},
  journal = {IEEE Trans. Inf. Theory},
  volume = {60},
  number = {9},
  pages = {5553--5572},
  month = sep,
  year = {2014}
}

@article{shannon1948,
  author = {Claude E. Shannon},
  title = {A mathematical theory of communication},
  journal = {Bell Syst. Tech. J.},
  volume = {27},
  pages = {379--423, 623--656},
  month = jul # " and " # oct,
  year = {1948}
}

@inproceedings{shannon1959,
  author = {Claude E. Shannon},
  title = {Coding theorems for a discrete source with a fidelity criterion},
  booktitle = {IRE Int. Conv. Rec., Part 4},
  volume = {7},
  pages = {142--163},
  year = {1959}
}

@inproceedings{strassen1962,
  author = {Volker Strassen},
  title = {{Asymptotische Absch{\"a}tzungen in Shannons Informationstheorie}},
  booktitle = {Trans. Third Prague Conf. Inf. Theory},
  pages = {689--723},
  editor = {Jaroslav Ko{\v z}e{\v s}n{\'i}k},
  year = {1964},
  address = {Prague},
}

@article{sulagastpar2020,
  author = {Erixhen Sula and Michael Gastpar},
  title = {The {Gray--Wyner} network and {Wyner}'s common information for {Gaussian} sources},
  journal = {IEEE Trans. Inf. Theory},
  volume = {68},
  number = {2},
  pages = {1369--1384},
  month = feb,
  year = {2022}
}

@article{tankosut2014,
  author = {Vincent Y. F. Tan and Oliver Kosut},
  title = {On the dispersions of three network information theory problems},
  journal = {IEEE Trans. Inf. Theory},
  volume = {60},
  number = {2},
  pages = {881--903},
  month = feb,
  year = {2014}
}

@article{watanabe2017,
  author = {Shun Watanabe},
  title = {Second-order region for {Gray--Wyner} network},
  journal = {IEEE Trans. Inf. Theory},
  volume = {63},
  number = {2},
  pages = {1006--1018},
  month = feb,
  year = {2017}
}

@article{watanabekuzuokatan2015,
  author = {Shun Watanabe and Shigeaki Kuzuoka and Vincent Y. F. Tan},
  title = {Nonasymptotic and second-order achievability bounds for coding with side-information},
  journal = {IEEE Trans. Inf. Theory},
  volume = {61},
  number = {4},
  pages = {1574--1605},
  month = apr,
  year = {2015}
}

@article{yu2023,
  author = {Lei Yu},
  title = {{Gray--Wyner} and mutual information regions for doubly symmetric binary sources and {Gaussian} sources},
  journal = {IEEE Trans. Inf. Theory},
  volume = {69},
  number = {10},
  pages = {6251--6268},
  month = oct,
  year = {2023}
}

@article{zhoumotani2023,
  author = {Lin Zhou and Mehul Motani},
  title = {Finite blocklength lossy source coding for discrete memoryless sources},
  journal = {Found. Trends Commun. Inf. Theory},
  volume = {20},
  number = {3},
  pages = {157--389},
  month = jun,
  year = {2023}
}

@article{zhoutanmotani2017gw,
  author = {Lin Zhou and Vincent Y. F. Tan and Mehul Motani},
  title = {Discrete lossy {Gray--Wyner} revisited: Second-order asymptotics, large and moderate deviations},
  journal = {IEEE Trans. Inf. Theory},
  volume = {63},
  number = {3},
  pages = {1766--1791},
  month = mar,
  year = {2017}
}

@article{zhoutanmotani2017sr,
  author = {Lin Zhou and Vincent Y. F. Tan and Mehul Motani},
  title = {Second-order and moderate deviations asymptotics for successive refinement},
  journal = {IEEE Trans. Inf. Theory},
  volume = {63},
  number = {5},
  pages = {2896--2921},
  month = may,
  year = {2017}
}

@misc{zhouwz2026,
  author = {Lin Zhou},
  title = {Exact second-order asymptotics for the discrete {Wyner--Ziv} problem},
  month = oct,
  year = {2026},
  note = {arXiv:2610.11998, version 1}
}

@article{kontoyiannisverdu2014,
  author = {Ioannis Kontoyiannis and Sergio Verd{\'u}},
  title = {Optimal lossless data compression: Non-asymptotics and asymptotics},
  journal = {IEEE Trans. Inf. Theory},
  volume = {60},
  number = {2},
  pages = {777--795},
  month = feb,
  year = {2014}
}

@article{hassanzadeh2020,
  author = {Parisa Hassanzadeh and Antonia M. Tulino and Jaime Llorca and Elza Erkip},
  title = {Rate-memory trade-off for caching and delivery of correlated sources},
  journal = {IEEE Trans. Inf. Theory},
  volume = {66},
  number = {4},
  pages = {2219--2251},
  month = apr,
  year = {2020}
}

@article{wyner1975,
  author = {Aaron D. Wyner},
  title = {The common information of two dependent random variables},
  journal = {IEEE Trans. Inf. Theory},
  volume = {21},
  number = {2},
  pages = {163--179},
  month = mar,
  year = {1975}
}

@article{viswanathaakyolrose2014,
  author = {Kumar B. Viswanatha and Emrah Akyol and Kenneth Rose},
  title = {The lossy common information of correlated sources},
  journal = {IEEE Trans. Inf. Theory},
  volume = {60},
  number = {6},
  pages = {3238--3253},
  month = jun,
  year = {2014}
}

@article{csiszar1998,
  author = {Imre Csisz{\'a}r},
  title = {The method of types},
  journal = {IEEE Trans. Inf. Theory},
  volume = {44},
  number = {6},
  pages = {2505--2523},
  month = oct,
  year = {1998}
}

@article{rimoldi1994,
  author = {Bixio Rimoldi},
  title = {Successive refinement of information: characterization of the achievable rates},
  journal = {IEEE Trans. Inf. Theory},
  volume = {40},
  number = {1},
  pages = {253--259},
  month = jan,
  year = {1994},
  doi = {10.1109/18.272493}
}

@book{bhatia1997,
  author = {Rajendra Bhatia},
  title = {Matrix Analysis},
  series = {Graduate Texts in Mathematics},
  volume = {169},
  publisher = {Springer},
  address = {New York, NY},
  year = {1997},
  doi = {10.1007/978-1-4612-0653-8}
}

@article{yavaskostinaeffros2024,
  author = {Recep Can Yavas and Victoria Kostina and Michelle Effros},
  title = {Third-order analysis of channel coding in the small-to-moderate deviations regime},
  journal = {IEEE Trans. Inf. Theory},
  volume = {70},
  number = {9},
  pages = {6139--6170},
  month = sep,
  year = {2024}
}

@article{wyner1978generalsources,
  author = {Aaron D. Wyner},
  title = {The rate-distortion function for source coding with side information at the decoder--{II}: General sources},
  journal = {Inf. Control},
  volume = {38},
  number = {1},
  pages = {60--80},
  month = jul,
  year = {1978}
}

@article{slepianwolf1973,
  author = {David Slepian and Jack K. Wolf},
  title = {Noiseless coding of correlated information sources},
  journal = {IEEE Trans. Inf. Theory},
  volume = {19},
  number = {4},
  pages = {471--480},
  month = jul,
  year = {1973}
}

@article{wynerziv1976,
  author = {Aaron D. Wyner and Jacob Ziv},
  title = {The rate-distortion function for source coding with side information at the decoder},
  journal = {IEEE Trans. Inf. Theory},
  volume = {22},
  number = {1},
  pages = {1--10},
  month = jan,
  year = {1976}
}

@book{coverthomas2006,
  author = {Thomas M. Cover and Joy A. Thomas},
  title = {Elements of Information Theory},
  edition = {2nd},
  publisher = {Wiley-Interscience},
  address = {Hoboken, NJ, USA},
  year = {2006}
}
\end{document}